\documentclass[a4paper, 11pt]{article}
\pdfoutput=1
\usepackage{amsmath, amsfonts, amscd, amssymb,  amsthm, mathrsfs}
\usepackage{longtable, geometry}
\usepackage[english]{babel}
\usepackage[utf8]{inputenc}
\usepackage[dvipdfmx]{graphicx}
\usepackage{bbm}
\usepackage{enumitem}
\usepackage{stmaryrd}
\usepackage{nicefrac}
\usepackage{bm}
\usepackage[svgnames,psnames]{xcolor}
\usepackage[colorlinks,citecolor=DarkGreen,linkcolor=FireBrick,linktocpage,unicode]{hyperref}
\usepackage{braket}
\usepackage{url, hyperref}

\hypersetup{  urlcolor=blue}

\usepackage{comment}
\usepackage{authblk}
\newtheorem{Thm}{Theorem}[section]
\newtheorem{Def}[Thm]{Definition}
\newtheorem{Lemm}[Thm]{Lemma}
\newtheorem{Prop}[Thm]{Proposition}
\newtheorem{Coro}[Thm]{Corollary}
\theoremstyle{definition}
\newtheorem*{Proof}{Proof}
\newtheorem{Rem}[Thm]{Remark}

\newcommand{\be}{\begin{equation}}
\newcommand{\ee}{\end{equation}}

\newcommand{\ba}{\begin{align}}
\newcommand{\ea}{\end{align}}

\newcommand{\ben}{\begin{equation*}}
\newcommand{\een}{\end{equation*}}

\def\i<#1>{\langle #1 \rangle}
\def\l<#1>{\left\langle #1 \right\rangle}
\def\b<#1>{\big\langle #1 \big\rangle}
\def\wrt{\ \text{w.r.t.}\ }
\newcommand{\la}{\langle}
\newcommand{\ra}{\rangle}
\newcommand{\bs}{\boldsymbol}

\newcommand{\Tr}{\mathrm{Tr}}

\newcommand{\BbbR}{\mathbb{R}}
\newcommand{\BbbN}{\mathbb{N}}

\newcommand{\BbbC}{\mathbb{C}}
\newcommand{\bq}{{\boldsymbol q}}

\newcommand{\vepsilon}{\varepsilon}
\newcommand{\vphi}{\varphi}
\newcommand{\no}{\nonumber \\}
\newcommand{\Ex}{\mathbb{E}}

\newcommand{\one}{{\mathchoice {\rm 1\mskip-4mu l} {\rm 1\mskip-4mu l}
{\rm 1\mskip-4.5mu l} {\rm 1\mskip-5mu l}}}
\newcommand{\up}{\uparrow}
\newcommand{\down}{\downarrow}

\newcommand{\hd}{\hat{d}}
\newcommand{\hf}{\hat{f}}
\newcommand{\hnd}{\hat{n}^d}
\newcommand{\hnf}{\hat{n}^f}

\newcommand{\ohnf}{\overline{\hat{n}}^f}
\newcommand{\vLa}{\varLambda}
\newcommand{\vPhi}{\varPhi}
\newcommand{\Np}{N_{\mathrm{p}}}
\newcommand{\im}{\mathrm{i}}
\newcommand{\h}{\mathfrak{H}}
\newcommand{\p}{\mathfrak{P}}

\newcommand{\vP}{\varPhi}
\newcommand{\VP}{{\bs \varPhi}}
\newcommand{\F}{\mathfrak{F}}
\newcommand{\X}{\mathfrak{X}}
\newcommand{\hn}{\hat{n}}

\makeatletter
  
  \@addtoreset{equation}{section}
\makeatother

\title{ \sf
Ground state properties of the periodic Anderson model with electron-phonon interactions

}

\date{}

\author[1]{Tadahiro Miyao\thanks{miyao@math.sci.hokudai.ac.jp}}
\author[1]{Hayato Tominaga\thanks{htominaga@frontier.hokudai.ac.jp}}

\affil[1]{Department of Mathematics,  Hokkaido University

Sapporo 060-0810,  Japan}

\begin{document}

\maketitle

\begin{abstract}
The periodic Anderson model (PAM) is a fundamental model describing heavy fermion systems' behavior. In this paper, we investigate the PAM in the presence of electron-phonon interactions. By utilizing a novel analytical methodology based on operator inequalities, we demonstrate that the ground state at half-filling is unique and in a spin-singlet state. Additionally, we establish that the ground state exhibits short-range antiferromagnetism.
\end{abstract}

\section{Introduction}\label{Sec1}

\subsection{Background}
Heavy fermion systems are a prototypical example of strongly correlated electron systems, exhibiting a plethora of phenomena such as unconventional superconductivity and  heavy effective masses.
The periodic Anderson model (PAM) is one of the most fundamental models for describing heavy fermion systems.
There is a substantial corpus of literature on the PAM, comprising numerical analyses; see, e.g., \cite{PhysRevB.85.235116, PhysRevB.93.235143, PhysRevB.48.10320, PhysRevLett.122.227201} and references cited therein.
Despite being relatively scarce, rigorous studies of the PAM have also been undertaken and have made substantial contributions to the understanding of heavy fermion systems; see, for instance, \cite{PhysRevB.72.075130, NOCE2006173, PhysRevB.65.212303, PhysRevLett.70.833, PhysRevB.50.6246, PhysRevB.58.7612, PhysRevLett.68.1030}.
In  \cite{PhysRevLett.68.1030}, Ueda, Tsunetsugu, and Sigrist show that the ground state of the symmetric PAM is unique and a spin singlet.
Their proof is based on the spin reflection positivity established by Lieb in his examination of the ground state of the Hubbard model \cite{Lieb1989}; Tian also employs the concept of spin reflection positivity to show that the ground state exhibits short-range antiferromagnetism \cite{PhysRevB.50.6246}. These rigorous results remain a firm cornerstone of subsequent studies on the PAM.
\medskip

This paper aims to rigorously  examine the effects of interactions between phonons and many-electron systems as characterized by the PAM.
The PAM, incorporating electron-phonon interactions, has been employed in the theoretical explication of unconventional superconductivity in heavy fermion systems \cite{PhysRevB.36.180}.  For recent studies, see  \cite{PhysRevB.99.155147, 
MITSUMOTO2007419, 
PhysRevB.87.121102} 
and references therein.
However, there has been a scarcity of rigorous examination of such models.
A more comprehensive depiction of our findings is as follows: This study concentrates on the interaction between conduction electrons and phonons and the interaction between localized electrons and phonons; in the scenario of either interaction, if the electron-phonon interaction is not particularly strong, we prove that the ground state at half-filling is unique and a spin singlet, and exhibits short-range antiferromagnetism.
\medskip

The methodological novelty of this paper is outlined below. In  \cite{PhysRevB.50.6246, PhysRevLett.68.1030}, the ground-state properties of the PAM are clarified by utilizing the method of the spin reflection positivity to the PAM.
  The concept of spin reflection positivity stems from the spatial reflection positivity of axiomatic quantum field theory \cite{Osterwalder1973, Osterwalder1975}. 
   It is a powerful analytical method that can be applied not only to the PAM but also to various models describing many-electron systems\footnote{See \cite{Shen1998,Tian2004} for a comprehensive review of the spin reflection positivity.
Also, refer to \cite{Tasaki2020}  for an instructive explanation of Lieb's theorem. For a recent development, see \cite{Yoshida2021}.}.
 Freericks and Lieb were the pioneers in applying the spin reflection positivity to electron-phonon interacting systems \cite{Freericks1995}; they succeeded in analyzing the ground-state properties of a class of general models, including the Holstein model.
It is well-established that the Lang-Firsov transformation is a crucial element in the examination of electron-phonon interacting systems, however, the method of \cite{Freericks1995} is incompatible with this transformation.
This obstacle was surmounted by introducing a fresh analytical approach based on operator inequalities in \cite{Miyao2016}. The theory of operator inequalities presented here  differs from those found in the standard textbooks of functional analysis, and has thus far been demonstrated to be highly efficacious in analyzing various models of many-electron systems, see, e.g., \cite{Miyao2012, Miyao2016, Miyao2019, MIYAO2021168467, Miyao2022}.
It should be noted that, despite not being widely acknowledged, the rigorous analysis of the PAM is more convoluted than the rigorous analysis of other models of many-electron systems.
Furthermore, the operators describing phonons are typically unbounded operators on infinite-dimensional Hilbert space, rendering mathematical treatment more intricate. As such, the endeavor to establish the uniqueness of the ground state of the PAM with electron-phonon interactions through existing methods is a highly daunting task. In this paper, we tackle this challenge by refining the analytical method based on the operator inequalities of \cite{Miyao2012, Miyao2016, MIYAO2021168467}.  It should be emphasized that the method presented in this paper is capable of analyzing a broad class of models.

\subsection{Ground state properties of the PAM}
To contextualize the importance of our results, we first provide an overview of the results for the standard PAM.
The Hamiltonian of the  periodic Anderson  model on a finite lattice $\vLa$ is given by
\begin{align}
H_\mathrm{PAM}
=&\sum_{x,  y\in \vLa}\sum_{\sigma=\uparrow,  \downarrow} (-t_{x,  y})d_{x,  \sigma}^* d_{y,  \sigma}
+\sum_{x\in \vLa}\sum_ {\sigma=\up, \down}\vepsilon_f n_{x, \sigma}^f \no
&+V\sum_{x\in \vLa}\sum_ {\sigma=\up, \down}(f_{x,  \sigma}^* d_{x,  \sigma}+d^*_{x,  \sigma} f_{x,  \sigma})
+U^f\sum_{x\in \vLa} n_{x,  \uparrow}^fn_{x,  \downarrow}^f.
\end{align}
Here, $d_{x, \sigma}$ and $f_{x, \sigma}$ are annihilation operators of conduction electrons and localized electrons, respectively, and satisfy the following anti-commutation relations:
\begin{align}
&\{d_{x, \sigma}, d_{y, \tau}\}=0=\{f_{x, \sigma}, f_{y, \tau}\}, \label{CAR1}\\
&\{d_{x, \sigma}, d_{y, \tau}^*\}=\delta_{x, y}\delta_{\sigma, \tau},\ \ \{f_{x, \sigma}, f^*_{y, \tau}\}=\delta_{x, y} \delta_{\sigma, \tau}. \label{CAR2}
\end{align}
$n_{x, \sigma}^f$ is the number operator of $f$-electrons at site $x$, defined by 
$n_{x, \sigma}^f=f_{x, \sigma}^* f_{x, \sigma}$.
$H_{\rm PAM}$ acts in the fermionic Fock space:
\be
\F_{\rm e}=\bigoplus_{n=0}^{4|\vLa|} \bigwedge^n (\mathfrak{h}\oplus \mathfrak{h}),
\quad \mathfrak{h}=
\ell^2(\vLa)\oplus \ell^2(\vLa), \label{DefFe}
\ee
where $\bigwedge^n$ denotes  the $n$-fold antisymmetric tensor product with $\bigwedge^0 (\mathfrak{h}\oplus \mathfrak{h})=\BbbC$.\footnote{Refer to \cite{Arai2016, Bratteli1997}  for mathematical definitions and basic properties of fermionic Fock spaces and annihilation operators.}
Throughout this paper, we will focus on the half-filling case. Therefore, in what follows, we consider $H_{\rm PAM}$ to be restricted to the following subspace:
\be
\F_{{\rm e}, 2|\vLa|}= \bigwedge^{2|\vLa|} (\mathfrak{h}\oplus \mathfrak{h}).
\ee
The hopping matrix element is denoted by $t_{x, y}$;  $\vepsilon_f$ represents  a local
potential; 
$U^f$ is the on-site interaction between spin-up
and spin-down electrons on the localized orbital, and $V$ is the
conduction-localized orbital hybridization.

Across the course of this paper, we assume the following.
\begin{description}
\item[\hypertarget{A1}{(A. 1)}] 
The parameters fulfill the following conditions:
\begin{itemize}
\item[(i)] $t_{x, y}\in \BbbR$ and $t_{x, y} =t_{y, x} $ for every $x, y \in \vLa$.
\item[(ii)] $\vepsilon_f\in \BbbR$, $U^f \in \BbbR$ and  $V \in \BbbR$ with $V\neq 0$.
\end{itemize}
\end{description}
Under these conditions, we see that $H_{\rm PAM}$ is self-adjoint.

Let $G_d=(\vLa, E)$ be the graph generated by the hopping matrix:
	$E=\{\{x, y\} : t_{x, y} \neq 0\}$ defines the set of edges.  The following assumption is essential to this paper.
\begin{description}
	\item[(A. 2)]\hypertarget{A2}{} $G_d$ is connected and bipartite. To be precise, 
	\begin{itemize}
	\item[(i)] for any $x, y\in \vLa $,  there is a sequence $\{\{x_i, x_{i+1}\}\}_{i=0}^{n-1}$  in  $E$ satisfying $x_0 = x,  x_n = y$ and $t_{x_i, x_{i+1}} \neq 0\ (i=0, 1, \dots, n-1)$; 
	\item[(ii)]
	there exists a partition $\vLa=\vLa_1\cup \vLa_2\ (\vLa_1\cap \vLa_2=\varnothing)$ of $\vLa$ satisfying $t_{x, y}=0$ when $x, y\in\vLa_1$ or $x, y\in\vLa_2$.
	
	\end{itemize}
	\end{description}
Next, let us define some spin operators. The spin operators ${\bs S}_x^d=(S_x^{d, (1)}, S_x^{d, (2)}, S_x^{d, (3)})$ of the conduction electrons at site $x$ are defined to be
\be
S_x^{d, (i)}=\frac{1}{2} \sum_{\sigma, \sigma\rq{}=\up, \down} d_{x, \sigma}^* (s^{(i)})_{\sigma, \sigma\rq{}}d_{x, \sigma},\quad i=1,2,3,
\ee
where $s^{(i)}\ (i=1, 2, 3)$ are the Pauli matrices:
\be
s^{(1)}=\begin{pmatrix}
0 & 1 \\
1 & 0
\end{pmatrix}, \ \ 
s^{(2)}=\begin{pmatrix}
0 & -\im \no
\im & 0
\end{pmatrix},\ \ 
s^{(3)}=\begin{pmatrix}
1 & 0 \no
0 & -1
\end{pmatrix}.
\ee
$(s^{(i)})_{\sigma, \sigma\rq{}}$ represents the matrix elements of $s^{(i)}$, with the correspondence $\up=1,\down=2$.
Under this convention, for example, $(s^{(1)})_{\up, \up}=(s^{(1)})_{1, 1}=0$ and $ (s^{(1)})_{\up, \down}=(s^{(1)})_{1, 2}=1$.
The spin operators ${\bs S}_x^f=(S_x^{f, (1)}, S_x^{f, (2)}, S_x^{f, (3)})$ of the $f$-electron at site $x$ are defined similarly: 
\be
S_x^{f, (i)}=\frac{1}{2} \sum_{\sigma, \sigma\rq{}=\up, \down} f_{x, \sigma}^* (s^{(i)})_{\sigma, \sigma\rq{}}f_{x, \sigma},\quad i=1,2,3.
\ee
The total spin operators ${\bs S}_{\rm tot}=(S_{\rm tot}^{(1)}, S_{\rm tot}^{(2)}, S_{\rm tot}^{(3)})$ are defined as 
\be
S_{\rm tot}^{(i)}=\sum_{x\in \vLa}(S_x^{d, (i)}+S_x^{f, (i)}),\quad i=1, 2, 3.
\ee
The Casimir operator is denoted by ${\bs S}_{\rm tot}^2$:
\be
{\bs S}_{\rm tot}^2=(S_{\rm tot}^{(1)})^2+(S_{\rm tot}^{(2)})^2+(S_{\rm tot}^{(3)})^2.
\ee
If the state $\vphi\in \F_{{\rm e}, 2|\vLa|}$ is an eigenvalue of ${\bs S}_{\rm tot}^2$ with  ${\bs S}_{\rm tot}^2 \vphi=S(S+1) \vphi$, then we say that $\vphi$ has total spin $S$.

The following theorem is proved in \cite{PhysRevLett.68.1030}:
\begin{Thm}\label{PAMTh} Assume \hyperlink{A1}{\bf (A. 1)} and \hyperlink{A2}{\bf (A. 2)}. In addition, assume that 
\be
U^f>0, \quad \vepsilon_f=-U^f/2. \label{Symm}
\ee
Then the ground state of $H_{\rm PAM}$ is unique and has total spin $S=0$.
\end{Thm}
When the condition \eqref{Symm} is satisfied, $H_{\rm PAM}$ is called the symmetric PAM, in particular. Under this condition, it is easily seen that $H_{\rm PAM}$ has hole-particle symmetry.

To state the next result, we define the ladder operators by 
\begin{align}
S_x^{d, (\pm)}=S_x^{d, (1)} \pm \im S_x^{d, (2)},\quad S_x^{f, (\pm)}=S_x^{f, (1)} \pm \im S_x^{f, (2)}.
\end{align}
The following theorem is proved in \cite{PhysRevB.50.6246}:

\begin{Thm}\label{PAMzigzag}
Under the same assumptions as in Theorem \ref{PAMTh}, let $\psi_{\rm g}$ be the ground state of $H_{\rm PAM}$.  For a given operator $A$, let $\la A\ra$ denote the ground state expectation of $A$ : $\la A\ra=\la\psi_{\rm g} \ket{A \psi_{\rm g}}$. Then, for any $x, y\in \vLa$, the following hold: 
\begin{align}
\gamma_x \gamma_y \big\la S_x^{d, (\pm)} S_y^{d, (\mp)}\big\ra&>0,\quad  \gamma_x \gamma_y \big\la S_x^{f, (\pm)} S_y^{f, (\mp)}\big\ra>0,\\ 
\gamma_x \gamma_y \big\la S_x^{d, (\pm)} S_y^{f, (\mp)}\big\ra&>0, \quad \gamma_x \gamma_y \big\la S_x^{f, (\pm)} S_y^{d, (\mp)}\big\ra>0,
\end{align}
where, $\gamma_x=1$ if $x\in \vLa_1$, $\gamma_x=-1$ if $x\in \vLa_2$.
\end{Thm}

Theorem \ref{PAMzigzag} implies that the ground state of $H_{\rm PAM}$ exhibits short-range antiferromagnetism.

\subsection{Main results}

In this paper, we examine the following two Hamiltonians that account for interactions between electrons and phonons:
\begin{description}
\item[Model 1:] 
\begin{align}
\boldsymbol{H}_d
=H_\mathrm{PAM} +U^d\sum_{x\in\Lambda}n^d_{x, \up}n_{x, \down}^d 
+g \sum_{x\in\Lambda} n_x^d(b_x^*+b_x) +\omega_0\Np.
\end{align}

\item[Model 2:]
\begin{align}
\boldsymbol{H}_f
=H_\mathrm{PAM} +U^d\sum_{x\in\Lambda}n^d_{x, \up}n_{x, \down}^d
+g \sum_{x\in\Lambda} n_x^f(b_x^*+b_x) +\omega_0\Np.
\end{align}

\end{description}

$\bs H_d$ and $\bs H_f$
 act in the Hilbert space:
\be
\F_{\rm e, 2|\vLa|}\otimes \h_{\rm ph},
\ee
where 
$\h_{\rm ph}$ is the Hilbert space describing the phonon states and is given by 
\be
\quad \h_{\rm ph}=L^2(\BbbR^{|\vLa|}).
\ee
The number operators of the conduction electrons at site $x$ are defined by 
$n_{x, \sigma}=d^*_{x, \sigma} d_{x, \sigma}$ and $n_x^d=n^d_{x, \up}+n_{x, \down}^d$; 
also, $n_x^f=n_{x, \up}^f+n_{x, \down}^f$ is the number operator of localized electrons at site $x$.
$b_x^*$ and $b_x$ are the phonon creation and annihilation operators at site $x$, respectively, and satisfy the usual commutation  relations\footnote{
More precisely, these commutation relations should be interpreted as holding in an appropriate subspace, e.g., $\mathcal{S}(\BbbR^{|\vLa|})$, the Schwartz space.
}:
\be
[b_x, b_y^*]=\delta_{x, y},\quad [b_x, b_y]=0.
\ee
$\Np=\sum_{x\in \vLa}b_x^*b_x$ is the phonon number operator.
$U^d$ is the strength of the Coulomb repulsion between conduction electrons; 
$g$  is  the coupling constant between phonons and  electrons; the
phonons are assumed to be dispersionless  with energy $\omega_0$.
$\bs H_d$ is a model that considers the interaction between conduction electrons and phonons.
On the other hand, $\bs H_f$ incorporates the effect of interaction between localized electrons and phonons.
More general interactions are discussed in Subsection \ref{Disc}.

In this paper, we assume the following:
\begin{description}
\item[\hypertarget{A3}{(A. 3)}] 
 $g\in\mathbb{R}, U^d\in \BbbR$ and $\omega_0>0$.
\end{description}
Using Kato--Rellich's theorem \cite[Theorem X.12]{Reed1975}, we see that both $\bs H_d$ and $\bs H_f$ are self-adjoint on $\mathrm{dom}(\Np)$ and bounded below.

To state the main results, we define the effective Coulomb energies between electrons as
\be
U_{{\rm eff}}^d=U^d- \frac{2g^2}{\omega_0},
\quad  U_{{\rm eff}}^f=U^f- \frac{2g^2}{\omega_0}.
\ee

The first main result of this paper is as follows:
\begin{Thm}\label{MainTh1}Assume \hyperlink{A1}{\bf (A. 1)}, \hyperlink{A2}{\bf (A. 2)} and  \hyperlink{A3}{\bf (A. 3)}. 
The following {\rm (i)} and {\rm (ii)} hold true:

\begin{itemize}
\item[\rm (i)]Assume that 
\begin{align}
U_{\rm eff}^d\ge 0,\quad  U^f>0, \quad
\varepsilon_f=\frac{1}{2}(U^d-U^f)-2\omega_0^{-1}g^2.
\end{align}
Then, the ground state of $\bs H_d$ is unique and  has total spin $S=0$.
\item[\rm (ii)] Assume that 
\begin{align}
U^d\ge 0,\quad
U_{\rm eff}^f>  0,\ \
\varepsilon_f=\frac{1}{2}(U^d-U^f)+2\omega_0^{-1}g^2.
\end{align}
Then, the ground state of $\bs H_f$ is unique and  has total spin $S=0$.
\end{itemize}
\end{Thm}

\begin{Rem}
\upshape
\begin{itemize}
\item[1. ] In (i), we can take $U^d_{\rm eff}=0$, while we cannot take $U^f=0$.
If $U^d_{\rm eff}>0$ and $U^f>0$, the theorem can be proved relatively easily by the method of \cite{Miyao2016, Miyao2018}.\footnote{More in detail, $\bs H_d$ is equivalent to the Holstein--Hubbard Hamiltonian on the enlarged lattice $\vLa \sqcup \vLa$, so the method of \cite{Miyao2016, Miyao2018} can be applied.} It is important to note that the method developed in this paper encompasses the case of $U^d_{\rm eff}=0$, which is mathematically challenging to analyze.
Similar observations hold true for (ii).
\item[2. ]
The condition $U^d_{\rm eff} \ge 0$ is equivalent to $|g|\le \sqrt{U^d \omega_0/2}$ .This suggests that if the interaction between conduction electrons and phonons is not excessively strong, the magnetic properties of the ground state of $H_{\rm PAM}$ stated in Theorem \ref{PAMTh} remain unaltered and stable. Analogously, a similar interpretation can be applied to (ii).
\end{itemize}
\end{Rem}

The second main result is the following theorem concerning the magnetic structure of the ground state:
\begin{Thm}\label{MainTh2}
Assume \hyperlink{A1}{\bf (A. 1)}, \hyperlink{A2}{\bf (A. 2)} and  \hyperlink{A3}{\bf (A. 3)}. 
For a given operator $A$, we denote by $\la A\ra_d$ the expectation concerning the ground state of $\bs H_d$.
Similarly, we denote by $\la A\ra_f$ the expectation concerning the ground state of $\bs H_f$.
The following {\rm (i)} and {\rm (ii)} hold true:
\begin{itemize}
\item[\rm (i)]
Under the same assumptions as in  {\rm (i)} of Theorem \ref{MainTh1}, we have
\begin{align}
\gamma_x \gamma_y \big\la S_x^{d, (\pm)} S_y^{d, (\mp)}\big\ra_d&>0,\quad  \gamma_x \gamma_y \big\la S_x^{f, (\pm)} S_y^{f, (\mp)}\big\ra_d>0,\\ 
\gamma_x \gamma_y \big\la S_x^{d, (\pm)} S_y^{f, (\mp)}\big\ra_d&>0, \quad \gamma_x \gamma_y \big\la S_x^{f, (\pm)} S_y^{d, (\mp)}\big\ra_d>0.
\end{align}
\item[\rm (ii)]Under the same assumptions as in  {\rm (ii)} of Theorem \ref{MainTh1}, we have
\begin{align}
\gamma_x \gamma_y \big\la S_x^{d, (\pm)} S_y^{d, (\mp)}\big\ra_f&>0,\quad  \gamma_x \gamma_y \big\la S_x^{f, (\pm)} S_y^{f, (\mp)}\big\ra_f>0,\\ 
\gamma_x \gamma_y \big\la S_x^{d, (\pm)} S_y^{f, (\mp)}\big\ra_f&>0, \quad \gamma_x \gamma_y \big\la S_x^{f, (\pm)} S_y^{d, (\mp)}\big\ra_f>0.
\end{align}
\end{itemize}
\end{Thm}

\begin{Rem}
\upshape 

Theorem \ref{MainTh2} implies that the short-range antiferromagnetism of the ground state of $H_{\rm PAM}$ described in Theorem \ref{PAMTh} is robustly stable when the electron-phonon interactions are not excessively strong.
\end{Rem}

\subsection{Discussion}\label{Disc}
The approach put forth in this paper is capable of addressing more comprehensive electron-phonon interactions. For instance, consider the following Hamiltonian:
\begin{align}
\boldsymbol{H}
=H_\mathrm{PAM} +\sum_{x, y\in\vLa}U^d_{x, y}n_{x, \up}^d n_{y, \down}^d 
+\sum_{x, y\in\vLa}g_{x, y}n_x^d(b_y^*+b_y) +\omega_0\Np.
\end{align}
By imposing suitable restrictions on $U^d_{x, y}$ and $g_{x, y}$, we can demonstrate analogous outcomes as those presented in Theorems \ref{MainTh1} and \ref{MainTh2}. Furthermore, the method presented in this paper can also be applied to examine systems that involve interactions between electrons and the quantized electromagnetic field. For prior examinations of systems composed of many electrons interacting with quantized electromagnetic fields, please refer to, among others, \cite{GIULIANI2012461, Miyao2019}.  It is worth noting that there has yet to be a thorough examination of the interaction between many electrons as characterized by the PAM and quantized electromagnetic fields.
\medskip

The stability of the magnetic properties of the ground state of $H_{\rm PAM}$ as outlined thus far can be coherently accounted for by the theory established in \cite{Miyao2019, Miyao2022}.  
In essence, the stability can be succinctly attributed to the fact that $H_{\rm PAM}, \bs H_d$ and $ \bs H_f$ belong to the Marshall--Lieb--Mattis stability class $\mathscr{A}_{\rm MLM}$ on the extended lattice $\vLa\sqcup \vLa$.
It has been demonstrated that Hamiltonians belonging to $\mathscr{A}_{\rm MLM}$ possess unique ground states with a total spin of $S=0$, which further exhibit short-range antiferromagnetism, as shown in Theorems \ref{PAMzigzag} and \ref{MainTh2}.
\medskip

In addition to the PAM, the Kondo lattice model (KLM) is a noteworthy instance of a Hamiltonian belonging to $\mathscr{A}_{\rm MLM}$.
The authors of \cite{MIYAO2021168467} conducted a comprehensive examination of the system where many electrons described by the KLM interact with phonons, and proved that the ground state exhibits similar characteristics to those outlined in Theorems \ref{PAMTh} and \ref{PAMzigzag}.\footnote{See, in particular, \cite[Example 1]{MIYAO2021168467}.}
This result aligns with expectations, as the KLM is derived from the PAM via the extended Kondo limit  \cite{PhysRevB.65.212303}:
 \be
 U^f\to \infty,\quad V\to \infty, \quad \frac{V^2}{U^f}\to \mbox{const.},\quad \vepsilon_f=-\frac{U^f}{2}.
 \ee
 \medskip
 
 Additionally, it has been established that besides class  $\mathscr{A}_{\rm MLM}$, a variety of other stability classes exist, indicating that the ground states of various Hamiltonians that describe many-electron systems display some shared properties. Nevertheless, the rich and diverse range of phenomena arising from interactions between multiple electrons makes it highly likely that numerous stability classes remain undiscovered. Their identification is crucial in achieving a deeper comprehension of many-electron systems.

\subsection{Organization}

The structure and organization of the present work are as follows. In Section \ref{Sec2}, we provide mathematical preliminaries by introducing operator inequalities and highlighting their fundamental properties. These operator inequalities serve as the analytical foundation of this study.
In Section \ref{Sec3}, we show that the ground state is unique and exhibits short-range antiferromagnetism under the assumption that the heat semigroup generated by the Hamiltonian is ergodic, as stated in Theorem  \ref{MainPI}.
The proof of Theorem \ref{MainPI}, which asserts the ergodicity of the heat semigroup, is a complex and nuanced task that is addressed in Sections \ref{Sec4}-\ref{Sec8}. Initially, in Section \ref{Sec4}, an abstract theorem outlining the proof is presented. This novel method facilitates examining interacting systems of electrons described by the PAM and phonons. The abstract theorem contains five assumptions. In order to apply the theorem in practice, we will verify the validity of these assumptions for the PAM in Sections \ref{Sec5}-\ref{Sec8}.
 Section \ref{Sec9} completes the proof of the main theorems by showing that the ground state has total spin $S=0$. Finally, in Appendices \ref{PfProjection3} and \ref{PfConnect}, we prove two crucial propositions that would interrupt the flow of the main argument because the proofs are too long.

\subsection*{Acknowledgements}

T.M. was supported by JSPS KAKENHI Grant Numbers 18K03315, 20KK0304.\\
\vspace{2mm}

\noindent
{\bf Data Availability}\\
 Data sharing not applicable to this article as no datasets were generated or analysed during
the current study.
\vspace{2mm}

\noindent
{\bf Financial interests}\\
 Author T.M. and H.T. declare they have no financial interests.
\section{Preliminaries}\label{Sec2}

\subsection{Basic definitions}
This section briefly explains the operator inequalities necessary to prove the main theorems. The operator inequalities introduced here are different from those in ordinary functional analysis textbooks and characterize the analytical approach of this paper.

First, basic terms related to the  operator inequalities will be introduced.
Let $\mathfrak{X}$ be a  complex separable Hilbert space.
We denote by $\mathscr{B}(\mathfrak{X})$ the Banach space of all bounded operators on $\mathfrak{X}$.

\begin{Def}\label{DefPPInq} \upshape
A {\it Hilbert cone} $\mathfrak{P}$  in $\mathfrak{X}$   is  a closed convex  cone  obeying:
\begin{itemize}
\item[(i)] $\i<u|v>\geq 0$ for every  $u, v\in \mathfrak{P}$;
\item[(ii)] for each  $w\in \mathfrak{X}$, there exist  $u,u',v,v'\in \mathfrak{P}$ such that   $w=u-v+\im (u'-v')$  and $\i<u|v>=\i<u'|v'>=0.$
\end{itemize}
A vector $ u \in\mathfrak{P}$ is said to be {\it  positive w.r.t.} $\mathfrak{P}$. We write this as $u \geq 0$ w.r.t. $\mathfrak{P}$. A vector $v \in\mathfrak{X}$ is called {\it strictly positive w.r.t.} $\mathfrak{P
}$,  whenever $\i<v|u>>0$ for all $ u \in \mathfrak{P}\setminus\{0\}$. We express  this as $v>0$ w.r.t. $\mathfrak{P}$.
\end{Def}

The operator inequalities introduced below form the basis of the analytical methods in this paper.
\begin{Def}\upshape 
Let $A\in\mathscr B(\mathfrak{X})$. 
\begin{itemize}
\item[(i)] $A$ is {\it  positivity preserving  w.r.t.}  $\mathfrak{P}$ if $A\mathfrak{P}\subseteq \mathfrak{P}$. We write this as $A\unrhd 0\wrt \mathfrak{P}$.
\item[(ii)] $A$ is {\it positivity improving  w.r.t.} $\mathfrak{P}$ if,  for each  $u \in \mathfrak{P
} \setminus \{0\},\ A u >0\wrt \mathfrak{P}$ holds.  We express  this as $A\rhd 0\wrt \mathfrak{P}$.
\end{itemize}
Remark that the notations of the operator inequalities are  borrowed  from \cite{Miura2003}.
\end{Def}

The following corollary provides  fundamental properties for practical applications of the operator inequalities defined above.
\begin{Lemm}\label{PPBasic}
Let $A,B\in\mathscr B(\mathfrak{X})$.
 Suppose that  $A\unrhd 0 $ and $B\unrhd0\wrt \mathfrak{P}$. We have the following {\rm (i)}--{\rm (iv)}:
 \begin{itemize}
 \item[\rm (i)] For every  $u, v\in \mathfrak{P}$,  $\la u|Av\ra\ge 0$ holds.
 \item[{\rm (ii)}] If $a\ge 0 $ and $ b\ge 0$, then $aA+bB\unrhd0\wrt \mathfrak{P}$.
 \item[\rm (iii)] $A^*\unrhd 0$ w.r.t. $\mathfrak{P}$.
 \item[{\rm  (iv)}] $AB\unrhd0\wrt \mathfrak{P}$.
 \end{itemize}
\end{Lemm}
\begin{proof}
See, e.g., \cite{Miura2003,Miyao2016}. 
\end{proof}

Arguments combining the operator inequalities  with  limit operations can be  justified by the following lemma: 

\begin{Lemm}\label{Wcl}
Let $\{A_n\}_{n=1}^{\infty }$ and $A$ be 
bounded operators on $\mathfrak{X}$. If $A_n\unrhd 0$ w.r.t. $\p$ and $A_n$ weakly converges to $A$ as $n\to \infty$, then $A\unrhd 0$ w.r.t. $\p$ holds.
\end{Lemm}
\begin{proof}
See, e.g., \cite[Proposition A.1]{Miyao2020}.
\end{proof}

The following lemma is useful in analyzing the ground state properties.
\begin{Lemm}\label{PPSP}
Assume that $A\in\mathscr{B}(\X)\ (A\neq 0)$ satisfies $A\unrhd0\wrt\p$. If $u\in \X$  satsfies $u>0\wrt\p$, then $Au\neq0$ holds.
\end{Lemm}
\begin{proof}
See, e.g., \cite[Theorem A.7]{Miyao2020}
\end{proof}

Let $\mathfrak{X}_{\mathbb{R}}$ be the real subspace of $\mathfrak{X}$ generated by $\mathfrak{P}$.  If $A\in \mathscr{B}(\mathfrak{X})$ satisfies $A\mathfrak{X}_{\mathbb{R}} \subseteq \mathfrak{X}_{\mathbb{R}}$, then we say that $A$ {\it preserves the reality w.r.t. $\mathfrak{P}$}. Note that $A$ preserves the reality w.r.t. $\p$ if and only if $\la u|Av\ra\in \BbbR$ for every $u, v\in \mathfrak{X}_{\BbbR}$.

\begin{Def} \upshape
Let $A, B\in\mathscr B(\mathfrak{X})$ be reality preserving w.r.t. $\mathfrak{P}$.  
If   $A-B\unrhd 0$ holds, then we write this as  
$A\unrhd B \wrt \mathfrak{P}$.
In what follows, we understand that $A$ and $B$ are always assumed to be reality preserving when one writes $A\unrhd B$ w.r.t. $\mathfrak{P}$.
\end{Def}

The following two lemmas are useful for practical applications:
\begin{Lemm}\label{InqSeki}
Let $A,B,C,D\in\mathscr B(\mathfrak X)$. Suppose $A\unrhd B\unrhd0\wrt\mathfrak{P}$ and $C\unrhd D\unrhd0\wrt\mathfrak P$. Then we have $AC\unrhd BD\unrhd0\wrt\mathfrak P$.
\end{Lemm}
\begin{proof}
For proof, see, e.g., \cite{Miura2003,Miyao2016}. 
\end{proof}

\begin{Lemm}\label{ppiexp1}
Let $A,B$ be self-adjoint operators on $\mathfrak{X}$. Assume that $A$ is bounded from below and    $B\in\mathscr B(\mathfrak{X})$. Furthermore, suppose that $e^{-tA}\unrhd0\wrt\mathfrak{P}$ for all $t\geq0$ and $B\unrhd0\wrt\mathfrak{P}$. Then we have $e^{-t(A-B)}\unrhd e^{-tA}\wrt\mathfrak{P}$ for all $t\geq0$.
\end{Lemm}
\begin{proof}
See, e.g., \cite[Theorem A.3]{Miyao2020}.
\end{proof}

\begin{Def}
\upshape
Let $A$ be a self-adjoint operator on $\mathfrak{X}$,   bounded from below.
The semigroup  generated by $A$, $\{e^{-tA}\}_{t\ge 0}$,  is said to be {\it ergodic  w.r.t.} $\mathfrak{P}$, if the following (i) and (ii) are satisfied:
\begin{itemize}
\item[(i)] $e^{-tA} \unrhd 0$ w.r.t. $\mathfrak{P}$ for all $t\ge 0$;
\item[(ii)] for each $u, v\in \mathfrak{P} \setminus \{0\}$, there exists a $t\ge 0$
such that $\langle u| e^{-tA} v\rangle >0$. Note that $t$ could depend on $u$ and $v$.
\end{itemize}

\end{Def}

The following lemma   immediately follows  from the definitions:
\begin{Lemm}
Let $A$ be a self-adjoint operator on $\mathfrak{X}$, bounded from below. If $e^{-tA}\rhd 0$ w.r.t. $\mathfrak{P}$ for all $t>0$, then the semigroup  $\{e^{-tA}\}_{t\ge 0}$ is ergodic w.r.t. $\mathfrak{P}$.
\end{Lemm}

The following theorem is employed in showing the uniqueness of the ground state:
\begin{Thm}[Perron--Frobenius--Faris]\label{pff}
Let $A$ be a self-adjoint operator, bounded from below. 
Set $E(A)=\inf \mathrm{spec}(A)$, where $\mathrm{spec}(A)$ indicates the spectrum of $A$.
Assume that $E(A)$ is an eigenvalue of $A$.  If $\{e^{-tA}\}_{t\ge 0}$ is ergodic w.r.t. $\mathfrak{P}$, then $\dim \ker(A-E(A))=1$ and $\ker(A-E(A))$ is spanned by a   strictly positive vector w.r.t. $\mathfrak{P}$.
\end{Thm}

\begin{proof}
See, e.g.,  \cite{ deimling1985nonlinear, Faris1972}.
\end{proof}

\subsection{Fiber direct integral of Hilbert cones}

This subsection summarizes the basic properties of the fiber direct integral of Hilbert cones. For details, see \cite{Bs1976}.

Let $\mathfrak{X}$ be a complex Hilbert space and let $(M, \mathfrak{M}, \mu)$ be a $\sigma$-finite measure space.
The Hilbert space of  $L^2(M, d\mu; \mathfrak{X})$ of square integrable $\mathfrak{X}$-valued functions 
is called a {\it constant fiber direct integral}, and is written as 
$
\int^{\oplus}_M\mathfrak{X} d\mu
$ \cite[Section XIII.16]{Reed1978}.
The inner product on $\int^{\oplus}_M\mathfrak{X} d\mu$ is  given by 
$
\la \Phi|\Psi\ra=\int_M\la \Phi(m)|\Psi(m)\ra_{\mathfrak{X} }d\mu,
$
where $\la \cdot |\cdot \ra_{\mathfrak{X}}$ is the inner product on $\mathfrak{X}$.
As is well-known,   $L^2(M, d\mu; \mathfrak{X})$ can be naturally identified with $\mathfrak{X} \otimes L^2(M, d\mu)$:
\begin{align}
\mathfrak{X}\otimes L^2(M, d\mu)=\int_M^{\oplus} \mathfrak{X} d\mu. \label{TensorIdn}
\end{align}

We denote by 
$L^{\infty}(M, d\mu; \mathscr{B}(\mathfrak{X}))$  the space of measurable functions from $M$ to $\mathscr{B}(\mathfrak{X})$ with the norm: 
\begin{align}
\|A\|_{\infty}=\mathrm{ess.sup} \|A(m)\|_{\mathscr{B}(\mathfrak{X})}.
\end{align}
A bounded operator $A$ on $\int^{\oplus}_M\mathfrak{X} d\mu$ is said to be decomposed by the direct integral decomposition, if and only if there is a function $A(\cdot)\in L^{\infty}(M, d\mu; \mathscr{B}(\mathfrak{X}))$ such that 
\begin{align}
(A\Psi)(m)=A(m)\Psi(m),\quad  \Psi\in \int^{\oplus}_M\mathfrak{X} d\mu. \label{Decomp1}
\end{align}
In this case, we call $A$ {\it decomposable} and write 
\begin{align}
A=\int^{\oplus}_MA(m)d\mu.
\end{align}

The following simple lemma is frequently used in applications:
\begin{Lemm}
Let $B\in \mathscr{B}(\mathfrak{X})$. Under the identification (\ref{TensorIdn}), we have
\begin{align}
B\otimes \mathbbm{1}=\int^{\oplus}_MBd\mu. \label{BOI}
\end{align}
\end{Lemm}

Given a Hilbert cone  $\p $  in $\mathfrak{X}$, we set 
\begin{align}
\int^{\oplus}_M \p  d\mu=\bigg\{
\Psi\in \int^{\oplus}_M \mathfrak{X} d\mu\, :\, \mbox{$\Psi(m) \ge 0$ w.r.t. $\p $ for $\mu$-a.e.} 
\bigg\}.
\end{align}
It is not hard to check that $ \int^{\oplus}_M\p  d\mu$ is a Hilbert cone cone in $ \int^{\oplus}_M\mathfrak{X} d\mu$.
We call $\int^{\oplus}_M\p  d\mu$ a {\it  direct integral of $\p $}.

\begin{Prop}\label{BasicDP}
Let $A=\int^{\oplus}_MA(m) d\mu$ be a decomposable operator on $\int^{\oplus}_M\mathfrak{X} d\mu$.
If $A(m) \unrhd 0$ w.r.t. $\p $ for $\mu$-a.e., then
$A\unrhd 0$ w.r.t. $\int^{\oplus}_M\p  d\mu$.
\end{Prop}
\begin{proof}For each $\Psi\in \int^{\oplus} _M\p  d\mu$, we have
$
(A\Psi)(m)=A(m) \Psi(m)\ge 0 $ w.r.t. $\p $ for $\mu$-a.e..
Hence, $A\Psi \ge 0$ w.r.t. $\int^{\oplus}_M\p  d\mu$. 
\end{proof}

The following basic proposition is often useful:
\begin{Prop}\label{BasicDP2}
Under the identification (\ref{TensorIdn}), we have the following:
\begin{itemize}
\item[\rm (i)] Let $B\in \mathscr{B}(\mathfrak{X})$. If $B\unrhd 0$ w.r.t. $\p $, then
$B\otimes \mathbbm{1} \unrhd 0$ w.r.t. $\int^{\oplus}_M\p  d\mu$. 
\item[\rm (ii)] Let $C$ be a bounded linear operator on $L^2(M, d\mu)$. Let $L^2_+(M, d\mu)$  be the  Hilbert cone in $L^2(M, d\mu)$ given by  \be L^2_+(M, d\mu)=\{f\in L^2(M, d\mu) : f(m) \ge 0\ \mbox{a.e. $\mu$}\}.
\ee
If $C\unrhd 0$ w.r.t. $L^2(M, d\mu)_+$, then $\mathbbm{1}\otimes C \unrhd 0$ w.r.t. $\int^{\oplus}_M\p  d\mu$.
\end{itemize}
\end{Prop}
\begin{proof}
See, e.g., \cite[Corollary I.4, Proposition I.5]{Miyao2016(2)}.
\end{proof}

\subsection{Operator inequalities in $\mathscr{L}^2(\mathfrak{X})$ }

This subsection introduces a particular Hilbert cone, which is useful for studying many-electron systems and describes its fundamental properties.

Let  $\mathscr{L}^2(\mathfrak{X})$ be the set of all Hilbert--Schmidt operators on $\mathfrak{X}$:
 \be
 \mathscr{L}^2(\mathfrak{X})=\{\xi\in \mathscr{B}(\mathfrak{X})\, :\, \Tr[\xi^*\xi]<\infty\}.
 \ee
In what follows, we regard $\mathscr{L}^2(\mathfrak{X})$ as a Hilbert space equipped with the inner product 
$\langle \xi| \eta\rangle_2=\Tr[\xi^* \eta]\ (\xi, \eta\in \mathscr{L}^2(\mathfrak{X}))$.
We often abbreviate the inner product by omitting the subscript $2$ if no confusion arises.

Let $\vartheta$ be an antiunitary operator on $\mathfrak{X}$. We define the linear operatar $\Psi_\vartheta:\mathfrak{X}\otimes\mathfrak{X}\longrightarrow\mathscr L^2(\mathfrak X)$ by
\begin{align}
\Psi_\vartheta(\phi\otimes\vartheta\psi)=|\phi\rangle\langle\psi|,\quad\phi,\psi\in\mathfrak{X}. \label{DefTheta}
\end{align}
Since $\Psi_\vartheta$ is  unitary, we can identify $\mathfrak{X}\otimes\mathfrak{X}$ with $\mathscr{L}^2(\mathfrak{X})$, naturally. 
We shall express this identification as follows:
\begin{align}
\mathfrak{X}\otimes \mathfrak{X}\underset{\Psi_{\vartheta}}{=}\mathscr{L}^2(\mathfrak{X}). \label{IdnSym}
\end{align}
Occasionally, we abbreviate (\ref{IdnSym}) by omitting the subscript $\Psi_{\vartheta}$ if no confusion arises.

Given  $A\in\mathscr B(\mathfrak{X})$, we define the left multiplication operator, $\mathcal{L}(A)$,  and the right multiplication operator, $\mathcal{R}(A)$,  as follows:
\begin{align}
\mathcal{L}(A) \xi=A\xi,\ \ \mathcal{R}(A)\xi=\xi A,\ \ \xi\in \mathscr{L}^2(\mathfrak{X}).
\end{align}
We  readily confirm that $\mathcal{L}(A)$ and $\mathcal{R}(A)$ are bounded operators on $\mathscr{L}^2(\mathfrak{X})$ and satisfy
\begin{align}
\mathcal{L}(A)\mathcal{L}(B)=\mathcal{L}(AB),\ \ \ \mathcal{R}(A)\mathcal{R}(B) =\mathcal{R}(BA).
\end{align}
Under the identification (\ref{IdnSym}), we have
\begin{align}
A\otimes \mathbbm{1}=\mathcal L(A),\ \ \
\mathbbm{1}\otimes A = \mathcal R(\vartheta A^*\vartheta). \label{LRIden}
\end{align}
Set 
\begin{align}
\mathscr L^{2}_{+}(\mathfrak{X})=\{\xi \in\mathscr L^2(\mathfrak{X})\,:\,\xi\geq0\}, \label{DefI_{2, +}}
\end{align}
where the inequality in the right hand side of (\ref{DefI_{2, +}}) indicates the standard operator inequality.\footnote{To be precise,
$\xi\ge 0$ if and only if $\la x |\xi x\ra\ge 0$ for all $x, y\in \mathfrak{X}$.}
It is well-known that $\mathscr{L}^2_{+}(\mathfrak{X})$ is a Hilbert cone in $\mathscr{L}^2(\mathfrak{X})$, see, e.g.,  \cite[Proposition 2.5]{Miyao2016}. 
\begin{Def}\rm 
 We introduce the  Hilbert cone in $\mathfrak{X}\otimes \mathfrak{X}$ by  $\mathfrak C=\Psi_\vartheta^{-1}(\mathscr L^2_{+}(\mathfrak{X}))$. Taking the identification (\ref{IdnSym}) into account, we have the following identification:
 \begin{align}
 \mathfrak{C}=\mathscr{L}^2_{ +}(\mathfrak{X}).
 \end{align}
\end{Def}

The following proposition is fundamental to the analysis of this paper:

\begin{Prop}\label{PPI2}
Let $A\in \mathscr{B}(\mathfrak{X})$.
Then we have $\mathcal L(A)\mathcal R(A^*)\unrhd0\wrt \mathscr{L}^2_{ +}(\mathfrak{X})$. Hence, under  the identification (\ref{IdnSym}), we have $A\otimes \vartheta A \vartheta \unrhd 0$ w.r.t. $\mathfrak{C}$.
\end{Prop}
\begin{proof}
Take $\xi \in\mathscr{L}^2_{+}(\mathfrak{X})$, arbitrarily. Then we find that 
$
\mathcal L(A)\mathcal R(A^*) \xi= A\xi A^*\ge 0,
$
which implies that $\mathcal L(A)\mathcal R(A^*)\unrhd0\wrt \mathscr{L}^2_{ +}(\mathfrak{X})$. 
\end{proof}

Note that Proposition \ref{PPI2} abstracts from the idea of  the spin reflection positivity.

\begin{Coro}
Let $A\in \mathscr{B}(\mathfrak{X})$. Then
$
\exp( A\otimes \mathbbm{1}+\mathbbm{1}\otimes \vartheta A \vartheta)
\unrhd 0$ w.r.t. $\mathfrak{C}$.
\end{Coro}
\begin{proof}
 By using Proposition \ref{PPI2}, we have   
 $
\exp( A\otimes \mathbbm{1}+\mathbbm{1}\otimes \vartheta A \vartheta)=e^A\otimes \vartheta e^{A} \vartheta \unrhd 0$ w.r.t. $\mathfrak{C}$.
\end{proof}

\section{Uniqueness of ground state}\label{Sec3}

\subsection{
Main theorem in Section \ref{Sec3}}

First, we discuss $\bs H_d$ in detail. The basic strategy of the proof for the claim for $\bs H_f$ is the same (technically, it is simpler than the case of $\bs H_d$). For the readers' convenience, in Subection \ref{StHf}, we summarize the notes on the proof for the $\bs H_f$ case.

In the following, $\bs H_d$ is denoted as $\bs H$ unless there is a risk of confusion.
$\bs H$ commutes with the total spin operators $S_{\rm tot}^{(1)}, S_{\rm tot}^{(2)} $ and $S_{\rm tot}^{(3)}$ in the following strong sense: 
$[e^{-\beta H}, S_{\rm tot}^{(i)}]=0$ $(i=1,2,3,\ \beta \ge 0)$. So, we restrict $\bs H$ to the following subspace: 
\be
\h=\h_{\rm e}\otimes \h_{\rm ph},\quad
\h_{\rm e}=\ker (S_{\rm tot}^{(3)}) \cap \F_{{\rm e}, 2|\vLa|}.
\ee
$\h$ is called the $M=0$ subspace of $\F_{\rm e, 2|\vLa|}\otimes \h_{\rm ph}$.

The main theorem of this section is as follows.
\begin{Thm}\label{GSPart}
The ground state of ${\bs H}$ in $\h$ is unique.\footnote{In other words, the ground state of ${\bs H} \restriction \h$ is unique.}
 If we denote the expectation associated with this ground state by $\la \cdot \ra$, then the following hold:
\begin{align}
\gamma_x \gamma_y \big\la S_x^{d, (\pm)} S_y^{d, (\mp)}\big\ra&>0,\quad  \gamma_x \gamma_y \big\la S_x^{f, (\pm)} S_y^{f, (\mp)}\big\ra>0,\\ 
\gamma_x \gamma_y \big\la S_x^{d, (\pm)} S_y^{f, (\mp)}\big\ra&>0, \quad \gamma_x \gamma_y \big\la S_x^{f, (\pm)} S_y^{d, (\mp)}\big\ra>0.
\end{align}
\end{Thm}
The proof of Theorem \ref{GSPart} is quite complicated.
In the remainder of this section, we give the proof of this theorem, assuming that a particular theorem holds.
Then, in Sections \ref{Sec4}-\ref{Sec8}, we prove the theorem we have assumed.

\subsection{ Setting up proper Hilbert cones}

\subsubsection{Construction of the proper Hilbert cone in $\h_{\rm e}$}
In the remainder of this section, we will illustrate the proof flow of Theorem \ref{GSPart}. In proving the uniqueness of the ground state of Theorem \ref{GSPart}, we would like to apply Theorem \ref{pff}. In order to do so, we need to set up an appropriate reference Hilbert cone. In this subsection, we will construct the reference Hilbert cone.

First, we construct a Hilbert cone in the Hilbert space $\h_{\rm e}$ describing electrons.
For a given Hilbert space $\mathfrak{X}$, we denote by $\F(\X)$ the fermionic Fock space over $\X$:
\be
\F(\X)=\bigoplus_{n=0}^{\mathrm{dim} \X} \bigwedge^n\X,\quad \bigwedge^0 \X=\BbbC.
\ee
Then, $\F_{\rm e}$ given by $\eqref{DefFe}$ can be expressed as $\F_{\rm e}=\F(\mathfrak{h} \oplus \mathfrak{h})$.
Next, as a preparatory step, let us construct the following identification concerning the fermionic Fock spaces:
\be
\F(\mathfrak{h}\oplus \mathfrak{h})=\F(\mathfrak{h}) \otimes \F(\mathfrak{h}). \label{FIdn}
\ee
Let $\hat{d}_x$ and $\hat{f}_x$ be the annihilation operators in $\F(\mathfrak{h})$:
\begin{align}
\{\hat{d}_x, \hat{d}_y^*\}=\delta_{x, y}, \quad\{\hat{f}_x, \hat{f}_y^*\}=\delta_{x, y}, \quad \{\hat{d}_x, \hat{f}^*_y\}=0,\\
\{\hat{d}_x, \hat{d}_y\}=0, \quad\{\hat{f}_x, \hat{f}_y\}=0, \quad \{\hat{d}_x, \hat{f}_y\}=0.
\end{align}
In addition, let $\hat{N}$ be the number operator in $\F(\mathfrak{h})$: $\hat{N}=\sum_{x\in \vLa} \hat{d}_x^*\hat{d}_x+\sum_{x\in \vLa} \hat{f}_x^* \hat{f}_x$. We can then construct the unitary operator $\iota$ from $\F(\mathfrak{h}\oplus \mathfrak{h})$ to $\F(\mathfrak{h})\otimes \F(\mathfrak{h})$, which satisfies the following:
\begin{align}
\iota d_{x, \up} \iota^{-1}=\hd_x \otimes \one,\ \ \  \iota d_{x, \down}\iota^{-1}=(-1)^{\hat{N}} \otimes \hd_x, \ \ \ 
\iota f_{x, \up}\iota^{-1}=\hf_x \otimes \one,\ \ \  \iota f_{x, \down}\iota^{-1}=(-1)^{\hat{N}} \otimes \hf_x \label{AnniIdn}
\end{align}
and 
\be
\iota |\varnothing \ra_{\F(\mathfrak{h}\oplus \mathfrak{h})}=|\varnothing \ra_{\F(\mathfrak{h})} \otimes |\varnothing \ra_{\F(\mathfrak{h})},
\ee
where $|\varnothing \ra_{\F(\mathfrak{h}\oplus \mathfrak{h})}$ and $|\varnothing \ra_{\F(\mathfrak{h})}$
are the Fock vacuums in $\F(\mathfrak{h}\oplus \mathfrak{h})$ and $\F(\mathfrak{h})$, respectively.
In what follows, we do not explicitly specify the $\iota$ that gives the identification \eqref{FIdn}, unless there is a risk of confusion: for example, we may simply write as 
\be
 d_{x, \up} =\hd_x \otimes \one, \quad f_{x, \down}=(-1)^{\hat{N}} \otimes \hf_x, \quad  |\varnothing \ra_{\F(\mathfrak{h}\oplus \mathfrak{h})}=|\varnothing \ra_{\F(\mathfrak{h})} \otimes |\varnothing \ra_{\F(\mathfrak{h})}.
\ee

Let us examine how $\h_{\rm e}$ is expressed under the identification \eqref{FIdn}.
First, the  $N$-fermion subspace  $\F_N(\mathfrak{h} \oplus \mathfrak{h})=\bigwedge^N (\mathfrak{h} \oplus \mathfrak{h})$ is represented as 
\be
\F_N(\mathfrak{h} \oplus \mathfrak{h})=\bigoplus_{m+n=N}\F_m(\mathfrak{h}) \otimes \F_n(\mathfrak{h}),\quad
\F_m(\mathfrak{h})=\bigwedge^m \mathfrak{h}.
\ee
Therefore, $\h_{\rm e}=\F_{\mathrm{e}, 2|\vLa|} \cap \ker(S^{(3)}_{\rm tot})$ can be expressed as 
\be
\h_\mathrm{e}=\mathfrak{E}\otimes \mathfrak{E}, \quad\mathfrak{E}=\bigwedge^{|\vLa|} \big(
\ell^2(\vLa)\oplus \ell^2(\vLa)
\big). \label{IdnE}
\ee

\begin{Def}\rm
Suppose that 
$X_d\subseteq \vLa $ and $ X_f\subseteq \vLa$ satisfy $|X_d|+|X_f|=|\vLa|$. 
In this case, ${\bs X}=(X_d, X_f)\subset \vLa\times \vLa$ is called an {\it electron configuration}. 
We denote by $\mathscr{C}$ the set of all electron configurations:
\be
\mathscr{C}=\{
{\bs X}=(X_d, X_f)\, :\, X_d\subseteq \vLa,\ X_f\subseteq \vLa, |X_d|+|X_f|=|\vLa|
\}.
\ee\begin{figure}[t]
\centering
 \includegraphics[scale=0.8]{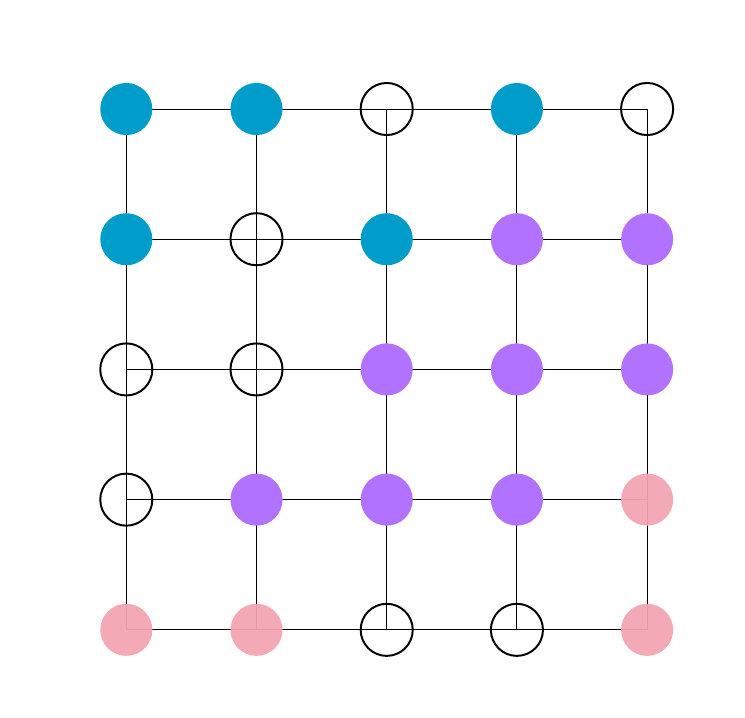}
\caption{The blue-colored sites are occupied by conduction electrons and correspond to $X_d$; the orange-colored sites are occupied by $f$-electrons and correspond to $X_f$. The purple-colored sites are occupied by both $f$-electrons and conduction electrons, corresponding to $X_d\cap X_f$.}
\label{Fig0}
\end{figure}
For a given ${\bs X}=(X_{d}, X_f)\in \mathscr{C}$, define the vector $|{\bs X}\ra$ in $\mathfrak{E}$ by
\begin{align}
\big|
{\bs X}\big\ra
=\Bigg[\prod_{x\in X_{d}} \hd_{x}^*\Bigg]
\Bigg[\prod_{x\in X_{f}} \hf_{x}^*\Bigg]
 |\varnothing\ra, \label{DefVecX}
\end{align}
where $|\varnothing\ra$ is the Fock vacuum; 
the elements of $\vLa$ are assumed to be numbered ($\vLa=\{x_1, x_2, \dots, x_{|\vLa|}\}$), and the order of the product $\prod_{x\in X}A_x\ (X\subseteq \vLa)$ is assumed to follow this order throughout this paper: 
\be
\prod_{x\in X}A_x=A_{x_{i_1}}A_{x_{i_2}} \cdots A_{x_{i_{|X|}}}\ (i_1<i_2<\cdots <i_{|X|}).
\ee
Fig. \ref{Fig0} depicts the vector $\ket{\bs X}$ using colors  for a two-dimensional square lattice.
Under the above setup, $\{|{\bs X}\ra\, :\, {\bs X} \in \mathscr{C}\}$ is a complete orthonormal system (CONS) of $\mathfrak{E}$.
\end{Def}

Next,  define the antiunitary operator $\vartheta: \mathfrak{E}\longrightarrow \mathfrak{E}$ by 
\begin{align}
\vartheta\Big(\sum_{\boldsymbol{X}\in \mathscr{C}} c_{\boldsymbol{X}}|\boldsymbol{X}\rangle\Big)
=\sum_{\boldsymbol{X}\in \mathscr{C}} c_{\boldsymbol{X}}^*|\boldsymbol{X}\rangle\quad (c_{\boldsymbol
{X}}\in\mathbb{C}).
\end{align}
With this $\vartheta$, we can define the  unitary operator $\Psi_{\vartheta}$ that gives the identification of $\h_{\rm e}$ with $\mathscr{L}^2(\mathfrak{E})$ in Section \ref{Sec2}: 
\begin{align}
\Psi_{\vartheta}(u \otimes \vartheta v)=|u\rangle \langle v| \quad (u,v\in\mathfrak{E}).
\end{align}

\begin{Def}\rm 
Define the Hilbert cone in $\h_{\rm e}$ as follows:
\begin{align}
\p_{\rm e}=\mathscr{L}^2_+(\mathfrak{E}).
\end{align}

\end{Def}

\subsubsection{Construction of the proper Hilbert cone in $\h_{\rm ph}$}
For each $x\in \vLa$, define the  self-adjoint operators, $p_x$ and $q_x$, by 
\begin{align}
p_x=\frac{\im }{\sqrt{2}}(\overline{b_x^*-b_x}), \ \ \  q_ x=\frac{1}{\sqrt{2}} (\overline{b_x^*+b_x}), \label{Defpq}
\end{align}
where $\overline{A}$ is the closure of $A$. $q_x$ is a multiplication operator in $L^2(\BbbR^{|\vLa|})$ and $p_x$ is equal to  the partial differential operator $-\im \partial/\partial q_x$; 
as is well-known, these operators satisfy the standard commutation relation:
$[q_x, p_y]=\im \delta_{x,y}$. 
In this paper, we choose the following as the Hilbert cone in $L^2(\BbbR^{|\vLa|})$:
\begin{align}
L^2_+(\mathbb{R}^{|\vLa|})=\{f\in L^2(\BbbR^{|\vLa|}) : f(\bs q) \ge 0\ \mbox {a.e. }\}.
\end{align}

\subsubsection{Construction of the proper Hilbert cone in $\h$}
According to \eqref{TensorIdn}, $\h$ can be expressed as follows: 
\be
\h=\int^{\oplus}_{\BbbR^{|\vLa|}} \h_{\rm e} d\bs q. \label{HeFiber}
\ee
Then, we can define the Hilbert cone in $\h$ as
\be
\p=\int^{\oplus}_{\BbbR^{|\vLa|}} \p_{\rm e} d\bq.
\ee
Note that $\p$ can also be represented  as:
\be
\p
=\overline{\mathrm{coni}}\Big\{\psi\otimes f \in\h\,:\,\psi\in\p_{\rm e},f\in L_+^2(\BbbR^{|\vLa|})\Big\}, \label{PExp2}
\ee
where $\overline{\mathrm{coni}}(S)$ indicates the closure of $\mathrm{coni}(S)$, the conical hull of $S$.
For the proof of \eqref{PExp2}, see \cite[Proposition D.1]{MIYAO2021168467}.

\subsection{Deformation of the Hamiltonian}\label{DefoHam}

In this subsection, we transform the Hamiltonian into a convenient form for analysis by applying the Lang--Firsov and hole-particle transformations.

Set
\begin{align}
L_d
=-\im \frac{\sqrt2 g}{\omega_0}\sum_{x\in\vLa }n_x^dp_x.
\end{align}
Then we readily confirm that 
\begin{align}
e^{\im \frac{\pi}{2}N_\mathrm{p}}q_xe^{-\im \frac{\pi}{2}N_\mathrm{p}}
&=p_x, \quad
e^{\im \frac{\pi}{2}N_\mathrm{p}}p_xe^{-\im \frac{\pi}{2}N_\mathrm{p}}
=-q_x, \label{LF1}\\
e^{L_d}d_{x,\sigma}e^{-L_d}
&=\exp\left(\im \frac{\sqrt2 g}{\omega_0}p_x\right) d_{x,\sigma}, \quad
e^{L_d}f_{x,\sigma}e^{-L_d}
=f_{x,\sigma}, \label{LF2}\\
e^{L_d}b_xe^{-L_d}
&=b_x-\frac{g}{\omega_0} n_x^d. \label{LF3}
\end{align}
The unitary operator $e^{L_d}$ is called the {\it Lang--Firsov transformation} which was first introduced in \cite{Lang1963}.

\begin{Lemm}\label{HLFT}
One obtains the following:
\begin{align}
&e^{\im \frac{\pi}{2}N_\mathrm{p}}e^{L_d} \boldsymbol{H} e^{-L_d}e^{-\im \frac{\pi}{2}N_\mathrm{p}}\no
&=T_\up(\VP) +T_\down(\VP) +(\varepsilon_f+\omega_0^{-1}g^2)\sum_{x\in\vLa } n_{x}^f +V_\up(\VP) +V_\down(\VP) +U^f\sum_{x\in \vLa } n_{x,  \uparrow}^fn_{x,  \downarrow}^f \no
&\quad\quad+U^d_{\mathrm{eff}}\sum_{x\in\vLa }n_{x, \up}^dn_{x, \down}^d+\omega_0\Np-2\omega_0^{-1}g^2|\vLa |,
\end{align}
where
\begin{align}
\varPhi_x
&=\frac{\sqrt2g}{\omega_0}q_x,\quad
\varPhi_{x,y}
=\varPhi_x-\varPhi_y, \\
T_\sigma(\pm \VP)
&=\sum_{x,  y\in \vLa } (-t_{x,  y})d_{x,  \sigma}^* d_{y,  \sigma}e^{\pm \im \vP_{x,y}}, \\
V_\sigma(\pm \VP)
&=V\sum_{x\in \vLa}\Big(f_{x,  \sigma}^* d_{x,  \sigma}e^{\mp \im \vP_x}+d^*_{x,  \sigma} f_{x,  \sigma}e^{\pm \im \vP_x}\Big). \label{DefVP}
\end{align}

\end{Lemm}

\begin{proof}
Direct calculation using 
\eqref{LF1}-\eqref{LF3}. Note that $\sum_x n_x^d=2|\vLa |-\sum_x n_x^f$ is employed.
\end{proof}

In the following, we assume that 
\begin{align}
\varepsilon_f
= \frac{1}{2}(U^d-U^f) -2\omega_0^{-1}g^2.\label{ChEf}
\end{align}

The {\it hole-particle transformation} is a unitary operator $W$ on $\h_{\rm e}$ that satisfies: 
\be
W^* d_{x, \up}W=d_{x, \up}, \quad W^* d_{x, \down}W=\gamma_x d_{x, \down}^*, \quad W^* f_{x, \up}W=f_{x, \up}, \quad W^* f_{x, \down}W=-\gamma_x f_{x, \down}^*, \label{WProp}
\ee
where $\gamma_x$ is given in Theorem \ref{PAMzigzag}.

\begin{Lemm}\label{HHPT}
Set
$\mathcal{U}=e^{-L_d}e^{-\im \frac{\pi}{2}N_\mathrm{p}}W$. 
Define the self-adjoint operator $H$ by 
 $H=\mathcal{U}^*\boldsymbol{H}\mathcal{U}+2\omega_0^{-1}g^2|\vLa|-U^d_{\rm eff}|\vLa |/2$. 
Then, $H$ can be represented as: 
\begin{align}
H
&=H_0 -R,
\end{align}
where
\begin{align}
H_0
&=T_\up(\VP) +T_\down(-\VP) +V_\up(\VP) +V_\down(-\VP) 
 +\frac{U^d_{\mathrm{eff}}}{2}\sum_{x\in \vLa}(n_{x, \up}^d+n_{x, \down}^d) +\omega_0\Np, \label{DefH_0} \\
R
&= \frac{U^f}{2}\sum_{x\in \vLa } n_{x,  \uparrow}^fn_{x,  \downarrow}^f +
\frac{U^f}{2}\sum_{x\in \vLa } (\mathbbm{1}
-n_{x,  \uparrow}^f)(\mathbbm{1}-n_{x,  \downarrow}^f)+U^d_{\mathrm{eff}}\sum_{x\in\vLa }  n_{x, \up}^dn_{x, \down}^d.
\end{align}
\end{Lemm}

\begin{proof}
We denote by $N_{\rm e}$ the total electron number operator: $N_{\rm e}=\sum_{x\in \vLa} \sum_{\sigma=\up, \down}(n_{x, \sigma}^d+n_{x, \sigma}^f)$. Put $\mu=-U_{\rm eff}^d/2$.
Noting that $N_{\rm e} \restriction \h=2|\vLa|$, we see that 
\be
\mathcal{U}^*\boldsymbol{H}\mathcal{U}+2\mu|\vLa|= 
\mathcal{U}^*(\boldsymbol{H}+\mu N_{\rm e})\mathcal{U}.
\ee
Using this, the condition \eqref{ChEf}, and the following equations, we obtain the desired claim:
\begin{align}
W^* \bigg( -\frac{1}{2}n_x^f +n_{x, \up}^fn_{x, \down}^f \bigg) W
&= -\frac{1}{2}n_{x, \up}^fn_{x, \down}^f  -\frac{1}{2}(\mathbbm{1}-n_{x, \up}^f)(\mathbbm{1}-n_{x, \down}^f), \\
W^* \bigg(-\frac{1}{2}n_x^d+n_{x, \up}^d n_{x, \down}^d\bigg)W&=\frac{1}{2} n_x^d-n_{x, \up}^d n_{x, \down}^d-\frac{1}{2}.
\end{align}
\end{proof}

\subsection{Strategy of the proof of Theorem \ref{GSPart} }

The following theorem is essential in the proof of Theorem \ref{GSPart}:
\begin{Thm}\label{MainPI}
The semigroup $\{e^{-\beta H}\}_{\beta \ge 0}$ is ergodic w.r.t. $\p$.
\end{Thm}
The proof of this theorem is involved and lengthy and will be given in Sections \ref{Sec4}-\ref{Sec8}.

Given that Theorem \ref{MainPI} holds, Theorem \ref{GSPart} can be proved as follows: 
According to Theorem \ref{pff}, the ground state of $H$ is unique. Let $\psi$ be the ground state of $H$, then $\psi$ can be chosen to be strictly positive with respect to $\p$.
Since  $\bs H=\mathcal{U} H\mathcal{U}^*+\mathrm{const}.$ holds by Lemma \ref{HHPT}, we know that $\psi_{\rm g}=\mathcal{U} \psi$ is the ground state of $\bs H$.
Using \eqref{AnniIdn} and  Proposition \ref{PPI2}, we find that 
\be
\gamma_x \gamma_y \mathcal{U}^*S_x^{d, (+)} S_y^{d, (-)}\mathcal{U}=S_x^{d, (+)} S_y^{d, (-)} \unrhd 0\ \mbox{w.r.t. $\p$}.
\ee
Because  $S_x^{d, (+)} S_y^{d, (-)} \psi\neq 0$ holds by 
Lemma \ref{PPSP}, we obtain 
\be
\gamma_x \gamma_y \big\la S_x^{d, (+)} S_y^{d, (-)}\big\ra=\gamma_x \gamma_y \la \psi|\mathcal{U}^*S_x^{d, (+)} S_y^{d, (-)}\mathcal{U} \psi\big\ra=\la \psi|S_x^{d, (+)} S_y^{d, (-)}\psi\big\ra>0.
\ee
We can prove the remaining claims by applying the similar method to the other two-point correlation functions.

\subsection{The proof strategy for the model in which localized electrons and phonons interact} \label{StHf}
In this subsection, we explain the strategy of the proof for $\bs H_f$. Most parts of the proof are the same as in the case of $\bs H_d$. The main change is in part (Subsection \ref{DefoHam}), where the Hamiltonian is transformed into a form that is convenient for analysis.
We will explain this part in some detail.

Let us introduce the Lang--Firsov transformation for localized electrons:
\begin{align}
L_f
=-\im \frac{\sqrt2 g}{\omega_0}\sum_{x\in\vLa }n_x^fp_x.
\end{align}
Then, the following lemma corresponds to Lemma \ref{HLFT}:
\begin{Lemm}\label{HLFT2}
One obtains 
\begin{align}
&e^{\im \frac{\pi}{2}N_\mathrm{p}}e^{L_f} \boldsymbol{H}_f e^{-L_f}e^{-\im \frac{\pi}{2}N_\mathrm{p}}\no
&=T_\up +T_\down+(\varepsilon_f-\omega_0^{-1}g^2)\sum_{x\in\vLa } n_{x}^f +V_\up(-\VP) +V_\down(-\VP) +U_{\rm eff}^f\sum_{x\in \vLa } n_{x,  \uparrow}^fn_{x,  \downarrow}^f \no
&\quad\quad+U^d\sum_{x\in\vLa }n_{x, \up}^dn_{x, \down}^d+\omega_0\Np, 
\end{align}
where
\begin{align}
T_\sigma
&=\sum_{x,  y\in \vLa } (-t_{x,  y})d_{x,  \sigma}^* d_{y,  \sigma}.
\end{align}
\end{Lemm}

The following lemma  corresponds to Lemma \ref{HHPT}:

\begin{Lemm}\label{HHPT2}
Choose
$
\varepsilon_f=\frac{1}{2}(U^d-U^f)+2\omega_0^{-1}g^2
$. 
Set 
$\mathcal{U}_f=e^{-L_f}e^{-\im \frac{\pi}{2}N_\mathrm{p}}W$.
We define the self-adjoint operator $H_f$ by 
 $H_f=\mathcal{U}_f^*\boldsymbol{H}_f\mathcal{U}_f-U^d|\vLa |$ 
 Then, $H_f$ can be represented as 
\begin{align}
H_f
&=H^f_0 -R_f,
\end{align}
where
\begin{align}
H^f_0
&=T_\up+T_\down +V_\up(-\VP) +V_\down(\VP) 
 +\frac{U^d}{2}\sum_{x\in \vLa}(n_{x, \up}^d+n_{x, \down}^d) +\omega_0\Np, \label{DefH_0^f} \\
R_f
&= \frac{U_{\rm eff}^f}{2}\sum_{x\in \vLa } n_{x,  \uparrow}^fn_{x,  \downarrow}^f +
\frac{U_{\rm eff}^f}{2}\sum_{x\in \vLa } (\mathbbm{1}
-n_{x,  \uparrow}^f)(\mathbbm{1}-n_{x,  \downarrow}^f)+U^d\sum_{x\in\vLa }  n_{x, \up}^dn_{x, \down}^d.
\end{align}
\end{Lemm}

\begin{Rem}\label{Reason}
\upshape
Comparing Lemmas \ref{HHPT} and \ref{HHPT2}, we find that the transformed Hamiltonian $H_f$ is more straightforward in structure and easier to analyze. In a more detailed description, the hopping term ($T_{\sigma}$) of the conduction electrons in the Hamiltonian $H_f$ does not contain any phonon-related operators. On the other hand, in Lemma \ref{HHPT}, the hopping term ($T_{\sigma}(\pm \VP)$) of the transformed Hamiltonian $H$ contains the operators concerning phonons. Therefore, the analysis of $\bs H_d$ is technically much more complicated. Therefore, most of this paper will discuss $\bs H=\bs H_d$ in detail.

\end{Rem}

\section{Structure of the proof of Theorem \ref{MainPI}}\label{Sec4}
In the previous section, we found that Theorem \ref{MainPI} is essential in proving Theorem \ref{GSPart}.
In this section, we will give a broad overview of the structure of the proof of Theorem \ref{MainPI}; 
the more intricate parts of the proof are discussed in detail in Sections \ref{Sec5}-\ref{Sec8}.
In this section, only $\bs H=\bs H_d$ will be discussed in detail. See Remark \ref{Reason} for the reason.

\subsection{Abstraction of the structure of the proof}
This subsection aims to prove Theorem \ref{ergodicity}, which abstractly expresses the structure of the proof of Theorem \ref{MainPI}.

Let $\X$ be a complex Hilbert space. Suppose that we are given a  certain  Hilbert cone $\p$ in $\X$.

\begin{Def}\label{Def succ}\upshape
Let $A$ be a self-adjoint operator on $\X$, bounded from below. Let $B\in\mathscr{B}(\X)$. Suppose that 
$B\unrhd0 \wrt\p $ and $e^{-tA}\unrhd0\wrt\p $ for all $t\geq0$.
 We express   $\{e^{-tA}\}_{t\geq0}\succeq B\wrt\p $ if, 
for any  $u, v\in\p $ satisfying $\i<u| Bv>>0$,
 there exists a $t\ge 0$
such that  $\i<u| e^{-tA}v>>0$.
\end{Def}

\begin{Rem}\label{Yowai}
\rm
If $e^{-tA} \unrhd B$ w.r.t. $\p $ for all $t\ge 0$, then we readily confirm that $\{e^{-tA}\}_{t\geq0}\succeq B\wrt\p $.
From this, the inequality ``$\succeq$'' can be regarded as a more generalized concept than ``$\unrhd $''.
\end{Rem}

The following lemma is helpful in applying the new inequality-like notion introduced in Definition \ref{Def succ}.
\begin{Lemm}\label{succ prod}
Let $A$ be a self-adjoint operator on $\X$, bounded from below. Let $B_1,\ldots,B_n\in\mathscr{B}(\X)$.
Suppose that $e^{-tA}\unrhd0\wrt\p $ for all $t\geq0$.  Suppose that $B_j\unrhd0\wrt\p\ (j=1, \dots, n)$.
If $\{e^{-tA}\}_{t\geq0}\succeq B_j\wrt\p \ (j=1, \dots, n)$ hold, then we obtain
\begin{align}
\{e^{-tA}\}_{t\geq0}\succeq B_1\cdots B_n\wrt\p.  \label{succ prod-1}
\end{align}
\end{Lemm}

\begin{proof}
We prove Lemma \ref{succ prod} by mathematical induction.

Since
$\{e^{-tA}\}_{t\geq0}\succeq B_1\wrt\p $, \eqref{succ prod-1} holds for  $n=1$.
Assume that \eqref{succ prod-1} holds for some $n$.
Set $C_n=B_1\cdots B_n$.
Because $\{e^{-tA}\}_{t\geq0}\succeq C_{n}\wrt\p $ holds,
for each  $u, v\in\p $ satisfying $\i<u|C_nB_{n+1}v>>0$,  there is a $t \ge 0$ such that 
$
\i<u| e^{-tA}B_{n+1}v>>0.
$
By using the assumption: $\{e^{-tA}\}_{t\geq0}\succeq B_{n+1}\wrt\p $, we see that  there exists an $s\ge 0$ satisfying 
$
\i<u| e^{-(t+s)A}v>>0,
$
which implies that $\{e^{-tA}\}_{t\geq0}\succeq C_nB_{n+1}\wrt\p$. Therefore, $\{e^{-tA}\}_{t\geq0}\succeq B_1\cdots B_n\wrt\p $ holds for any $n\in \BbbN$.
\end{proof}

\begin{Coro}\label{succ prod2}
Let $A$ be a self-adjoint operator on $\X$, bounded from below. Let $B_1,B_2\in\mathscr{B}(\X)$.
Suppose that $e^{-tA}\unrhd0\wrt\p $ for all $t\geq0$.  Suppose that $B_j\unrhd0\wrt\p\ (j=1,2)$.
If $\{e^{-tA}\}_{t\geq0}\succeq B_j\wrt\p \ (j=1, 2)$ hold, then, for every $\beta \ge 0$, we obtain
\begin{align}
\{e^{-tA}\}_{t\geq0}\succeq B_1e^{-\beta A}B_2 \wrt\p.
\end{align}
\end{Coro}

\begin{proof}From Definition \ref{Def succ}, it is evident that $\{e^{-tA}\}_{t\geq0}\succeq e^{-\beta A} \wrt\p $  holds.
Hence, by applying  Lemma \ref{succ prod}, we see that  $\{e^{-tA}\}_{t\geq0}\succeq B_1e^{-\beta A}B_2 \wrt\p $ holds.
\end{proof}

The main theorem of this subsection is as follows: 
\begin{Thm}\label{ergodicity}
Let $A$ be a self-adjoint operator on $\X$, bounded from below.
Assume that \be
e^{-tA}\unrhd0\ \wrt\p  \quad( t\geq0). \label{ECon1}
\ee
Let $\mathfrak{I}$ be a subset of $\p \setminus\{0\}$ satisfying the following:  for any $\varphi, \psi\in\mathfrak{I}$, there exists a $\beta>0$ such that
\begin{align}
\i<\varphi| e^{-\beta A}\psi>>0. \label{ECon2}
\end{align}
Assume that, for any $u\in\p \setminus\{0\}$, there exists a family $\{E_u(\beta, \beta\rq{}) : \beta>0, \beta\rq{}>0\}$ of operators satisfying the following   {\rm (i)}, {\rm (ii)}, and {\rm (iii)}:
\begin{itemize}
\item[\rm (i)]  $E_u(\beta, \beta')\unrhd0\wrt\p $ for all $\beta>0$ and $ \beta\rq{}>0$.
\item[\rm (ii)]  $\{e^{-t A}\}_{t\geq0}\succeq E_u(\beta, \beta')\wrt\p $  for all $\beta>0$ and $ \beta\rq{}>0$. 
\item[\rm (iii)] The limit 
\begin{align}
u_0
=\lim_{\beta\to+0}\lim_{\beta'\to+0} E_u(\beta, \beta') u
\end{align}
exists. In addition, there exists a $\varphi\in\mathfrak{I}$ satisfying $u_0 \ge \vphi $ w.r.t. $\p $.
\end{itemize}
Then $\{e^{-\beta A}\}_{\beta\geq0}$ is  ergodic  w.r.t. $\p $.
\end{Thm}

\begin{proof}
For any $u, v\in\p \setminus\{0\}$, there exist families of operators 
$\{E_u(\beta, \beta\rq{}) : \beta>0, \beta\rq{}>0\}$ and $\{E_v(\beta, \beta\rq{}) : \beta>0, \beta\rq{}>0\}$
satisfying the following (a), (b) and (c): (a)
$\{e^{-tA}\}_{t\geq0}\succeq E_u(\beta, \beta')$ and $ \{e^{-tA}\}_{t\geq0}\succeq E_v(\beta, \beta') \wrt\p $;
(b)  the following  limits exist: 
\begin{align}
u_0
=\lim_{\beta\to+0}\lim_{\beta'\to+0} E_u(\beta, \beta') u, \quad
v_0
=\lim_{\beta\to+0}\lim_{\beta'\to+0} E_v(\beta, \beta')v, 
\end{align}
 and  (c) there are $\varphi, \psi\in\mathfrak{I}$ satisfying 
 $u_0 \ge \vphi $ and $ v_0 \ge \psi$ w.r.t. $\p $.

For these $\vphi$ and $\psi$ in (c), from  the assumption \eqref{ECon2}, we can take a $\beta_0>0$ such that 
\begin{align}
\i<\varphi| e^{-\beta_0 A}\psi>>0. \label{ergodic-pos}
\end{align}
By using (b), we have
\begin{align}
\lim_{\beta\to+0} \lim_{\beta'\to+0}
\i<u| E_u(\beta, \beta')^* e^{-\beta_0 A} E_v(\beta, \beta')  v> =\i<u_0| e^{-\beta_0 A}v_0>.
\end{align}
Hence, for any $\varepsilon>0$, there is a $\beta_1>0$ such that, for all $0<\beta<\beta_1$, it holds that 
\begin{align}
\Big| \lim_{\beta'\to+0}\i<u| E_u(\beta, \beta')^* e^{-\beta_0 A} E_v(\beta, \beta')  v> -\i<u_0 | e^{-\beta_0 A}v_0> \Big| 
<\varepsilon.
\end{align}
Then we fix $\beta$ arbitrarily, satisfying $0<\beta<\beta_1$. 
For any  $\varepsilon'>0$, there exists a $\beta_1'(\beta)>0$ such that if $0<\beta'<\beta_1'(\beta)$, then it holds that 
\begin{align}
\Big| \lim_{\beta'\to+0} \i<u| E_u(\beta, \beta')^* e^{-\beta_0 A} E_v(\beta, \beta')  v> -\i<u| E_u(\beta, \beta')^* e^{-\beta_0 A} E_v(\beta, \beta')  v> \Big| 
<\varepsilon'.
\end{align}
Summing up the above, we have
\begin{align}
\Big| \i<u| E_u(\beta, \beta')^* e^{-\beta_0 A} E_v(\beta, \beta')  v> -\i<u_0 | e^{-\beta_0 A}v_0> \Big|
<\varepsilon +\varepsilon'.
\end{align}
Since $\varepsilon$ and $ \varepsilon'$ are arbitrary, from \eqref{ergodic-pos}, these can be chosen so that $\i<\varphi| e^{-\beta_0 A}\psi>>\varepsilon +\varepsilon'$.
Because $u_0 \ge \vphi $ and $ v_0\ge \psi$ \wrt $\p $, one obtains 
\begin{align}
\i<u| E_u(\beta, \beta')^* e^{-\beta_0 A} E_v(\beta, \beta')  v>
>\i<u_0 | e^{-\beta_0 A}v_0> -\varepsilon -\varepsilon'
\ge \i<\varphi| e^{-\beta_0 A}\psi> -\varepsilon -\varepsilon'
>0.
\end{align}
As $\{e^{-tA}\}_{t\geq0}\succeq E_u(\beta, \beta')^* e^{-\beta_0 A} E_v(\beta, \beta')\wrt\p $  holds due to Corollary \ref{succ prod2}, there exists a $t(\beta, \beta')\geq0$ such that 
$
\i<u| e^{-t(\beta, \beta')A}  v>>0,
$
which implies that $\{e^{-\beta A}\}_{\beta\geq0}$ is  ergodic  w.r.t. $\p $. 
\end{proof}

Before proceeding, let us clarify our strategy for the proof of Theorem \ref{MainPI}:

\subsubsection*{Strategy for the proof of Theorem \ref{MainPI}}
In Theorem \ref{ergodicity}, take $A=H$. Construct  appropriate $\mathfrak{I}$ and $E_u(\beta, \beta\rq{})$ that satisfy the conditions of Theorem \ref{ergodicity} and show that $\{e^{-\beta H}\}_{\beta \ge 0}$ is ergodic by applying the same theorem.

\subsection{Basic properties of electron configuration}

In this subsection, we give some basic definitions related to electron configurations, which are necessary to check the conditions to apply Theorem \ref{ergodicity} to $H$.
\subsubsection{Basic definitions}

\begin{Def}\label{DefNei}\rm 
Let ${\bs X}, {\bs Y}\in \mathscr{C}$. 
${\bs X}$ and $ {\bs Y}$ are said to be {\it adjacent}  if these satisfy either of the following two conditions:
\begin{itemize}
\item[\rm (i)] There exist $x, y\in \vLa\ (x\neq y)$ such that $X_d \triangle Y_d=\{x, y\}$ and $X_f=Y_f$, and furthemore $t_{x, y}\neq 0$, where $X_d \triangle Y_d$ represents the symmetric difference of $X_d$ and $Y_d$:
$X_d\triangle Y_d=(X_d\setminus Y_d) \cup (Y_d\setminus X_d)$.

\item[\rm (ii)] There exists an $x\in \vLa$ such that $X_d\triangle Y_d=\{x\}=X_f\triangle Y_f$.
\end{itemize}
A pair $\{\bs X, \bs Y\}$ is said to be an {\it edge} if $\bs X$ and $\bs Y$ are adjacent.
The set of all edges is denoted by $\mathscr{E}$. The graph defined by $\mathscr{G}=(\mathscr{C}, \mathscr{E})$ plays an essential role in the following discussion.
We say that $\{\bs X, \bs Y\}$ is a  {\it $d$-edge} (resp. {\it $(d, f)$-edge}) if it satisfies (i) (resp. (ii)).

A sequence ${\bs X}_1, \dots, {\bs X}_n$ consisting of elements of $\mathscr{C}$ is said to be a {\it path} connecting $\bs X_1$ and $\bs X_n$ if it satisfies $\{{\bs X}_i, {\bs X}_{i+1}\}\in \mathscr{E}\ (i=1, \dots, n-1)$. We denote such a path by ${\bs p}={\bs X}_1 {\bs X}_2\cdots{\bs X}_n$;   we refer to $n$ as the length of the path ${\bs p}$ and denote it by $|\bs p|$.

\end{Def}

The next proposition forms the basis for the following discussion.
\begin{Prop}\label{Connect}
The graph $\mathscr{G}$ is connected: for any $\bs X, \bs Y\in \mathscr{C}$, there exists a path connecting $\bs X$ and $\bs Y$.
\end{Prop}
The proof of Proposition \ref{Connect} is somewhat involved and is given in Appendix \ref{PfConnect}.

\subsubsection{Basic operators associated with electronic configurations}
Here we define several operators necessary for the proof of Theorem \ref{MainPI}.

For each $X\subseteq \vLa$ and $\sigma=\up, \down$, we set
\begin{align}
P_{X, \sigma}^d=\prod_{x\in  X} n_{x, \sigma}^d,\quad \overline{P}^d_{ X, \sigma}=\prod_{x\in  X} \overline{n}_{x, \sigma}^d, \quad
P^f_{ X, \sigma}=\prod_{x\in  X} n_{x, \sigma}^f,\quad \overline{P}^f_{  X, \sigma}=\prod_{x\in  X} \overline{n}_{x, \sigma}^f,  \label{DefPPPP}
\end{align}
where
\be
 \overline{n}_{x, \sigma}^d=\mathbbm{1}-n_{x, \sigma}^d,\quad
 \overline{n}_{x, \sigma}^f=\mathbbm{1}-n_{x, \sigma}^f.
\ee

\begin{Def} \upshape
\begin{itemize}
\item[\rm (i)]
For each ${\bs X}=(X_d, X_f)\in \mathscr{C}$, define
\begin{align}
P_{\boldsymbol{X}}
=\prod_{\sigma=\up, \down}P_{X_f, \sigma}^f \overline{P}^f_{\vLa\setminus X_f, \sigma},  \quad
Q_{\boldsymbol{X}}
=\prod_{\sigma=\up, \down} P^f_{X_d, \sigma} \overline{P}^f_{\vLa \setminus X_d, \sigma}.
\label{Def oPX-1}
\end{align}
\item[(ii)]
For each $s>0$ and $\boldsymbol{X}\in\mathscr{C}$, define
\begin{align}
F_s(\boldsymbol{X})=P_{\boldsymbol{X}} e^{-s H_0 } Q_{\boldsymbol{X}} e^{-s H_0} P_{\boldsymbol{X}}, 
\end{align}
where $H_0$ is given by \eqref{DefH_0}.

\end{itemize}
\end{Def}

\subsection{Structure of the proof of Theorem \ref{MainPI}}
\subsubsection{Five key propositions}
Here we give five propositions necessary to prove Theorem \ref{MainPI}. Each proposition corresponds to an assumption in Theorem \ref{ergodicity}.
The proofs of the propositions are rather lengthy and will be given in separate sections.

The following proposition is  fundamental:
\begin{Prop}\label{BasePP}
$e^{-\beta H} \unrhd 0$ w.r.t. $\p $ for all $\beta \ge 0$.
\end{Prop}
We prove Proposition \ref{BasePP}  in Section \ref{Sec5}.
Proposition \ref{BasePP} corresponds to the condition \eqref{ECon1} in Theorem \ref{ergodicity}.

The following particular electron configuration frequently appears in the following discussions.

\begin{Def}\rm 
We define the electron configuration  $\boldsymbol{F}=(F_d,F_f)\in \mathscr{C}$ by $F_d=\varnothing$ and $ F_f=\vLa$.
\end{Def}

For each $\bs X, \bs Y\in \mathscr{C}$ and $f\in L^2(\BbbR^{|\vLa|})$,  we set
\be
\ket{\bs X, \bs Y; f}=\ket{\bs X, \bs Y}\otimes f.
\ee

The following proposition corresponds to \eqref{ECon2} in Theorem \ref{ergodicity}:
\begin{Prop}\label{ISP}
For any $f, g\in L_+^2(\BbbR^{|\vLa|})\setminus \{0\}$, there exists a $\beta >0$ such that 
\be
\la \bs F, \bs F; f|e^{-\beta H} |\bs F, \bs F; g\ra>0.
\ee
\end{Prop}

Proposition \ref{ISP} is proved in Section \ref{Sec6}.
The following proposition corresponds to the condition (i) of Theorem \ref{ergodicity}.
\begin{Prop}\label{BP1}
For all $s \ge 0$ and ${\bs X} \in \mathscr{C}$, 
$F_s(\bs X)\unrhd 0$ w.r.t. $\p $ holds.
\end{Prop}

We prove Proposition \ref{BP1}  in  Section \ref{Sec5}.

In the proof of Theorem \ref{MainPI}, the following operator plays a crucial role: 
\begin{Def}\rm 
For each $\beta\ge 0, \beta\rq{}\ge 0$ and 
path ${\bs p}={\bs X}_1{\bs X}_2\cdots {\bs X}_n$, we define
\be
F_{\beta, \beta\rq{}}(\bs p)=F_{\beta\rq{}}({\bs X}_1) e^{-\beta H_0} F_{\beta\rq{}}({\bs X}_2)e^{-\beta H_0}\cdots e^{-\beta H_0} F_{\beta\rq{}}({\bs X}_n).
\ee
\end{Def}

Regarding $F_{\beta, \beta\rq{}}(\bs p)$, the following two propositions hold.

\begin{Prop}\label{BP2} For any path $\bs p$, $\beta\ge 0 $ and $ \beta\rq{}\ge 0$, we have
\be
\{e^{-t H}\}_{t\ge 0} \succeq F_{\beta, \beta\rq{}}(\bs p)\ \ \wrt \p.
\ee
\end{Prop}

The proof of Proposition \ref{BP2} is given in Section \ref{Sec7}. This proposition corresponds to the condition (ii) of Theorem \ref{ergodicity}.

Any $\vphi\in \p\setminus \{0\}$  can be expressed as follows: 
\be
\vphi=\sum_{\bs X, \bs Y\in \mathscr{C}} \ket{ \bs X, \bs Y}\otimes \vphi_{\bs X, \bs Y},\quad \vphi_{\bs X, \bs Y}\in L^2(\BbbR^{|\vLa|}). \label{vphiDec}
\ee
Note that $\vphi_{\bs X, \bs X}\ge 0$ w.r.t. $L_+^2(\BbbR^{|\vLa|})$ holds and there exists an $\bs X\in \mathscr{C}$ such that $\vphi_{\bs X, \bs X}\neq 0$\footnote{The reasoning is as follows: assume that $\vphi_{\bs X, \bs X}=0$ for all $\bs X\in \mathscr{C}$.
Under the identification \eqref{IdnE}, we have $\Tr_{\mathscr{L}^2(\mathfrak{E})}[\vphi]=0$, which leads to $\vphi=0$, where $\Tr_{\mathscr{L}^2(\mathfrak{E})}$ indicates the partical trace  with respect to ${\mathscr{L}^2(\mathfrak{E})}$. This contradicts with $\vphi\neq 0$. 
}.
Denote by $\bs X(\vphi)$ one of the $\bs X$'s such that $\vphi_{\bs X, \bs X} \neq 0$ and  $|\bs X|_{\triangle}$ is maximal,  where $|{\bs X}|_\triangle :=|X_d\triangle X_f|$.

\begin{Prop}\label{BP3}
For a given $\vphi\in \p \setminus\{0\}$, let $\bs p=\bs X_0{\bs X}_1\cdots {\bs X}_{n+1}$ be a path connecting $\bs F$ and $\bs X(\vphi)$. Then, there exist  some $c>0$ and $f\in L_+^2(\BbbR^{|\vLa|})\setminus \{0\}$ such that 
\be
\lim_{\beta\to+0}\lim_{\beta'\to+0} \beta^{-2n-2} \beta^{\rq{}-4d}F_{\beta, \beta\rq{}}(\bs p)\vphi=c \ket{{\bs F}, {\bs F}; f}, 
\ee
where $d=\sum_{i=1}^n|{\bs X}_i|_\triangle+|{\bs F}|_\triangle$.
\end{Prop}
The proof of Proposition \ref{BP3} is given in Section \ref{Sec8}.
From this proposition, we can show the condition (iii) of Theorem \ref{ergodicity}.
The proof of Proposition \ref{BP3} is the most demanding of the five propositions listed here.

\subsubsection{Proof of Theorem \ref{MainPI} given Propositions \ref{BasePP}, \ref{ISP}, \ref{BP1}, \ref{BP2} and \ref{BP3}}

Set $\mathfrak{I}=\big\{\alpha|\boldsymbol{F}, \boldsymbol{F}; f\rangle\,:\,\alpha>0,\ f\in L_+^2(\BbbR^{|\vLa|})\setminus \{0\}\big\}$. 
According to Proposition \ref{ISP}, for any $f, g\in L^2_+(\BbbR^{|\vLa|}) \setminus \{0\}$, it holds that 
\begin{align}
\i<\boldsymbol{F}, \boldsymbol{F}; f| e^{-\beta H} |\boldsymbol{F}, \boldsymbol{F}; g>>0\quad (\beta > 0).
\end{align}

For a given $\vphi \in \p \setminus\{0\}$, let  $\bs p$  be the path given in Proposition \ref{BP3}.
For each $\beta>0$ and $\beta\rq{}>0$, define
\be
E_{\vphi}(\beta, \beta')
=\beta^{-2n-2}\beta'^{-4d} F_{\beta, \beta\rq{}}(\bs p).
\ee

From Proposition \ref{BP1}, $E_{\vphi}(\beta, \beta') \unrhd 0$ w.r.t. $\p $ holds. According to  Proposition \ref{BP2},  it can be deduced that 
\be
\{e^{-t H}\}_{t\ge 0} \succeq E_{\vphi}(\beta, \beta') \quad \wrt \p.
\ee
Furthermore,  based on Proposition \ref{BP3}, there exist $c>0$ and $f\in L_+^2(\BbbR^{|\vLa|})\setminus \{0\}$ such that 
\be
\lim_{\beta\to+0}\lim_{\beta'\to+0} E_{\vphi}(\beta, \beta')\vphi=c \ket{{\bs F}, {\bs F}; f} \quad(c>0).
\ee

Consequently, it is evident that all the assumptions of  Theorem \ref{ergodicity} are valid (with $u_0=\vphi=c \ket{{\bs F}, {\bs F}; f}$).
Thus, based on the same theorem, we conclude that $\{e^{-\beta H}\}_{\beta\ge 0}$ is ergodic w.r.t. $\p $.
\qed

\section{Proofs of Propositions \ref{BasePP} and  \ref{BP1}}\label{Sec5}
\subsection{Proof of Proposition \ref{BasePP}}
The operators $R_0$ and $R_1$ are defiend as
\begin{align}
R_0&=\frac{U^f}{2}\sum_{x\in \vLa } n_{x,  \uparrow}^fn_{x,  \downarrow}^f +
\frac{U^f}{2}\sum_{x\in \vLa } (\mathbbm{1}
-n_{x,  \uparrow}^f)(\mathbbm{1}-n_{x,  \downarrow}^f),\label{DefR0}\\
R_1&=U_{\mathrm{eff}}^d\sum_{x\in\vLa }  n_{x, \up}^dn_{x, \down}^d.
\end{align}

\begin{Lemm}\label{RProp}
We have the following:
\begin{itemize}
\item[\rm (i)] $R_0\unrhd 0$ w.r.t. $\p$.
\item[\rm (ii)] $R_1\unrhd 0$ w.r.t. $\p$.
\end{itemize}
\end{Lemm}
\begin{proof}
(i) By using the identification \eqref{IdnE} and Proposition \ref{PPI2}, we have
\begin{align}
n_{x, \up}^fn_{x, \down}^f=\hn_x^f \otimes \vartheta \hn_x^f \vartheta \unrhd 0, \quad
(\mathbbm{1}-n_{x, \up}^f)(\mathbbm{1}-n_{x, \down}^f)=(\mathbbm{1}-\hn_x^f) \otimes \vartheta (\mathbbm{1}-\hn_x^f) \vartheta \unrhd 0\ 
 \mbox{w.r.t. $\p_{\rm e}$},\label{fnPP}
\end{align}
which implies that $R_0 \unrhd 0$ w.r.t. $\p_{\rm e}$. Hence, applying Proposition \ref{BasicDP2}, we obtain the desired assertion. 

(ii) can be proved in the similar way as (i).
\end{proof}

\begin{Lemm}\label{PPH_0}
$e^{-\beta H_0} \unrhd 0$ w.r.t. $\p$
for all $\beta \ge 0$.
\end{Lemm}

\begin{proof}
We split the operator $H_0$ as follows:
\be
H_0=H_1+\omega_0\Np,
\ee
where
\begin{align}
H_1
=T_\up(\VP) +T_\down(-\VP) +V_\up(\VP) +V_\down(-\VP) +\frac{ U^d_{\mathrm{eff}}}{2}
\sum_{x\in\vLa}(n_{x, \up}^d +n_{x, \down}^d). \label{DefH_1}
\end{align}
For each $\bs q\in \BbbR^{|\vLa|}$, we define the self-adjoint operators $\hat{H}_1(\pm \VP(\bs q))$ acting on $\mathfrak{E}$ by
\be
\hat{H}_1(\pm \VP(\bs q))
=\sum_{x,  y\in \Lambda}  (-t_{x,  y})e^{\pm \im \vP_{x, y}(\bs q)}\hd_{x}^* \hd_{y}
-V\sum_{x\in\Lambda}(e^{\mp \im \vP_x(\bs q)}\hf_{x}^* \hd_{x}+ e^{\pm \im \vP_x(\bs q)}\hd^*_{x} \hf_{x})
+
\frac{ U^d_{\mathrm{eff}}}{2}
\sum_{x\in\vLa} \hn_{x}^d, \label{DefhH_1}
\ee
where, for each $\bq\in \BbbR^{|\vLa|}$, $\vP_{x, y}(\bq)$ and $ \vP_x(\bq)$ are the values at $\bq$ of the functions $\vP_{x, y}$ and $ \vP_x$, respectively.
Because $\vartheta \hat{H_1}(\VP(\bs q)) \vartheta =\hat{H_1}(-\VP(\bs q))$ holds, we obtain
\be
H_1=\int^{\oplus}_{\BbbR^{|\vLa|}}\hat{H}_1(\VP(\bs q)) \otimes \mathbbm{1} d\bs q +
\int^{\oplus}_{\BbbR^{|\vLa|}} \mathbbm{1}\otimes \vartheta \hat{H}_1(\VP(\bs q)) \vartheta d\bs q
\ee
under the identification \eqref{HeFiber}. 
Hence, using Propositions \ref{BasicDP} and \ref{PPI2}, we have
\be
e^{-\beta H_1}=\int^{\oplus}_{\BbbR^{|\vLa|}} \Big(e^{-\beta \hat{H_1}(\VP(\bs q))} \Big)\otimes \Big(\vartheta e^{-\beta \hat{H_1}(\VP(\bs q))} \vartheta \Big) d\bq \unrhd 0\ \mbox{w.r.t. $\p$}.
\ee
On the other hand, since $e^{-\beta \Np} \unrhd 0$ w.r.t. $L^2_+(\BbbR^{|\vLa|})$ for all $\beta \ge 0$, it follows from  Proposition \ref{BasicDP2}  that $\mathbbm{1} \otimes e^{-\beta \Np} \unrhd 0$ w.r.t. $\p$ for all $\beta \ge 0$.
Combining the above considerations with Trotter--Kato\rq{}s product formula \cite[Theorem S.20]{Reed1981}, we find that 
\be
e^{-\beta H_0}=\lim_{n\to \infty}\Big(e^{-\beta H_1/n} e^{-\beta \omega_0 \Np/n}\Big)^n \unrhd 0  \ \mbox{w.r.t. $\p$}, 
\ee
where we have used Lemma \ref{Wcl}.
\end{proof}

Proposition \ref{BasePP} can be concluded from the following lemma:
\begin{Lemm}\label{MonoE}
For every $\beta \ge 0$, we have
$
e^{-\beta H}
\unrhd e^{-\beta H_0}\wrt \p.
$  In particular, it holds from Lemma \ref{PPH_0} that $e^{-\beta H}\unrhd 0$ w.r.t. $\p$ for all $\beta \ge 0$.
\end{Lemm}
\begin{proof}
From  Lemma \ref{RProp}, we know that $R=R_0+R_1\unrhd 0$ w.r.t. $\p$. Since $H=H_0-R$, we obtain the desired assertion by Lemmas  \ref{ppiexp1} and \ref{PPH_0}.
\end{proof}

This completes the proof of Proposition \ref{BasePP}. \qed

\subsection{Proof of Proposition \ref{BP1}}

For each $X\subseteq \vLa$, we define the orthogonal projections on $\mathfrak{E}$ by 
\begin{align}
\hat{P}_{X}^d=\prod_{x\in  X} \hn_{x}^d,\quad \overline{\hat{P}}^d_{ X}=\prod_{x\in  X} \overline{\hn}_{x}^d, \quad
\hat{P}^f_{ X}=\prod_{x\in  X} \hn_{x}^f,\quad \overline{\hat{P}}^f_{  X}=\prod_{x\in  X} \overline{\hn}_{x}^f, \label{DefPPPP2}
\end{align}
where 
\be
 \overline{\hn}_{x}^d=\mathbbm{1}-\hn_{x}^d,\quad
 \overline{\hn}_{x}^f=\mathbbm{1}-\hn_{x}^f.
\ee
We readily confirm that, for each $\bs X=(X_d, X_f)\in \mathscr{C}$, 
\be
P_{\bs X}=\hat{P}^f_{X_f} \overline{\hat{P}}^f_{  \vLa\setminus X_f} \otimes  \vartheta \hat{P}^f_{X_f} \overline{\hat{P}}^f_{  \vLa\setminus X_f} \vartheta,\quad
Q_{\bs X} =\hat{P}_{X_d}^f \overline{\hat{P}}^f_{  \vLa\setminus X_d} \otimes  \vartheta\hat{P}_{X_d}^f \overline{\hat{P}}^f_{  \vLa\setminus X_d} \vartheta.
\ee
From Proposition \ref{PPI2}, we know that $P_{\bs X} \unrhd 0$ and $Q_{\bs X} \unrhd 0$ w.r.t. $\p$. Additionally, it holds that  $e^{-s H_0} \unrhd 0$ w.r.t. $\p$ ($s\ge 0$) from Lemma \ref{PPH_0}. 
Combining the above considerations with Lemma \ref{PPBasic},
we obtain  $F_s(\bs X) \unrhd 0$ w.r.t. $\p$. \qed

\section{Proof of Proposition \ref{ISP}}\label{Sec6}

According to Duhamel's formula, one has 
\begin{align}
&\i<{\bs F}, {\bs F};  g| e^{-\beta H_0} |{\bs F}, {\bs F};  h> \no
&= \i<{\bs F}, {\bs F};  g| e^{-\beta \omega_0N_\mathrm{p}} |{\bs F}, {\bs F};  h> \no
&\quad +\sum_{n\geq1}(-\beta)^n \int_{0\leq s_1\leq\cdots\leq s_n\leq1} \i<{\bs F}, {\bs F};  g| 
H_1(s_1) \cdots H_1(s_n) e^{-\beta\omega_0N_\mathrm{p}}  |{\bs F}, {\bs F};  h> \,d^n \bs s \label{DuhaDuha1}
\end{align}
 for all $g, h\in L^2_+(\BbbR^{|\vLa|})\setminus\{0\}$, where
$H_1(s)=e^{-s \beta \omega_0 \Np} H_1e^{s \beta \omega_0 \Np}$; $H_1$ is given by \eqref{DefH_1}.  Note that 
the right-hand side of \eqref{DuhaDuha1} converges in the operator norm topology.

\begin{Lemm}\label{H_1(s)Inqs}
For eavery $g, h\in L^2_+(\BbbR^{|\vLa|})\setminus\{0\}$, one obtains the following:
\begin{itemize}
\item[\rm (i)] $\i<{\bs F}, {\bs F};  g| H_1(s) e^{-\beta \omega_0 \Np} |{\bs F}, {\bs F}; h>=0$.
\item[\rm (ii)]
If $0\le s_1\le s_2\le \cdots \le s_n\le 1$, then 
\be
\Big|\i<{\bs F}, {\bs F};  g| 
H_1(s_1) \cdots H_1(s_n) e^{-\beta\omega_0N_\mathrm{p}}  |{\bs F}, {\bs F};  h>
\Big| \le \alpha^n\i<g| e^{-\beta \omega_0N_\mathrm{p}}h> 
\ee
holds, where 
\begin{align}
\alpha
=2\sum_{x, y\in\vLa} |t_{x, y}| +2V|\vLa | +2\sum_{x, y\in\vLa} |U_{\mathrm{eff}, x, y}|.
\end{align}
\end{itemize}
\end{Lemm}
\begin{proof}
(i) For any $x, y\in \vLa$, it holds that $\bra{\bs F, \bs F} d_{x, \sigma}^*d_{y, \sigma}\ket{\bs F, \bs F}=0$.
Hence, for all  $g', h'\in L_+^2(\BbbR^{|\vLa|})$ , we have
\begin{align}
\bra{\bs F, \bs F; g\rq{}} T_{\sigma}(\pm \VP) \ket{\bs F, \bs F; h\rq{}}
=\sum_{x, y\in \vLa} (-t_{x, y})\bra{\bs F, \bs F} d_{x, \sigma}^*d_{y, \sigma}\ket{\bs F, \bs F} \bra{ g\rq{}}e^{\pm\im  \vP_{x, y}} \ket{h\rq{}}=0.
\end{align}
Similarly, we have
\be
\bra{\bs F, \bs F; g\rq{}} V(\pm \VP) \ket{\bs F, \bs F; h\rq{}}=0,\quad \bra{\bs F, \bs F; g\rq{}} R\ket{\bs F, \bs F; h\rq{}}=0.
\ee
To sum up the above, we have $\bra{\bs F, \bs F; g\rq{}}H_1\ket{\bs F, \bs F; h\rq{}}=0$. 
Therefore, choosing $g\rq{}=e^{-s \beta \omega_0\Np}g $ and $ h\rq{}=e^{-(1-s)\beta \omega_0\Np}h$, we conclude that 
\begin{align}
\i<{\bs F}, {\bs F};  g| H_1(s) e^{-\beta \omega_0 \Np} |{\bs F}, {\bs F}; h>
=\bra{\bs F, \bs F; g\rq{}}H_1\ket{\bs F, \bs F; h\rq{}}=0.
\end{align}

(ii) To avoid unnecessary complications, we prove the case $n=2$. The general $n$ case can be proved similarly.

Because $e^{-t\omega_0\Np} \unrhd 0$ w.r.t. $L^2_+(\BbbR^{|\vLa|})$ for all $t \ge 0$, we readily confirm that 
\be
\big|
e^{-t\omega_0\Np} f
\big| \le e^{-t\omega_0\Np} |f|\quad (f\in L^2(\BbbR^{|\vLa|})).  \label{NpAbsInq}
\ee
Let us illustrate the idea of the proof with the following term that appears when we expand  $I:=\i<{\bs F}, {\bs F};  g| 
H_1(s_1)  H_1(s_2) e^{-\beta\omega_0N_\mathrm{p}}  |{\bs F}, {\bs F};  h>$: 
\be
\i<{\bs F}, {\bs F};  g| 
T_{\up}(\VP)(s_1) V_{\down}(-\VP)(s_2) e^{-\beta\omega_0N_\mathrm{p}}  |{\bs F}, {\bs F};  h>. \label{ExTerm}
\ee
First, let us express this term as
\begin{align}
\eqref{ExTerm}
=&\sum_{x_1, y_1}\sum_{x_2} (-t_{x_1, y_1}) V \bra{\bs F, \bs F} d_{x_1, \up}^*d_{y_1, \up}d_{x_2, \down}^* f_{x_2, \down}\ket{\bs F, \bs F}\no
&\qquad\times \bra{g} e^{-s_1 \beta \omega_0\Np} e^{\im \vP_{x_1, y_1}} e^{-(s_2-s_1) \beta \omega_0\Np} e^{\im \vP_{x_2}}e^{-(1-s_2) \beta \omega_0\Np}  h\ra.
\end{align} 
Using the inequality \eqref{NpAbsInq} twice over, we obtain 
\begin{align}
\Big|
\bra{g} e^{-s_1 \beta \omega_0\Np} e^{\im \vP_{x_1, y_1}} e^{-(s_2-s_1) \beta \omega_0\Np} e^{\im \vP_{x_2}}e^{-(1-s_2) \beta \omega_0\Np} h\ra
\Big|\le \i<g| e^{-\beta \omega_0N_\mathrm{p}}h> , 
\end{align}
which leads to 
\be
|\eqref{ExTerm}| \le \bigg(\sum_{x_1, y_1}\sum_{x_2} |t_{x_1, y_1}| V \bigg)\i<g| e^{-\beta \omega_0N_\mathrm{p}}h>.
\ee
We can similarly evaluate the other terms that appear when we expand $I$. From this, we obtain the desired inequality when $n=2$.
\end{proof}

If $\beta< e^{-\alpha/2} $, then from Lemma \ref{H_1(s)Inqs} and \eqref{DuhaDuha1}, we obtain
\begin{align}
\i<{\bs F}, {\bs F}; g| e^{-\beta H_0} |{\bs F}, {\bs F};  h> 
&\geq \i<g| e^{-\beta \omega_0N_\mathrm{p}}h> -\beta^2\sum_{n\geq2}\frac{\alpha^n}{n!} \i<g| e^{-\beta \omega_0N_\mathrm{p}}h> \no
&\geq (1-\beta^2e^\alpha)\i<g| e^{-\beta \omega_0N_\mathrm{p}}h> >0,
\end{align}
where,  in the last inequality, we have  used the fact that $e^{-\beta \omega_0\Np} \rhd 0$ w.r.t. $L^2(\BbbR^{|\vLa|})_+$ for all $\beta>0$. Therefore, by applying Lemma \ref{MonoE},  we see that if    $0<\beta< e^{-\alpha/2} $,  then 
\begin{align}
\i<\psi| e^{-\beta H}\varphi>
\geq \i<{\bs F}, {\bs F};  g| e^{-\beta H_0} |{\bs F},  {\bs F}; h>
>0
\end{align}
holds.  This completes the proof of Proposition \ref{ISP}. \qed

\section{Proof of Proposition \ref{BP2}}\label{Sec7}

In order to prove Proposition \ref{BP2}, we provide some lemmas.

First, we prove the following lemma concerning $R_0$ defined by \eqref{DefR0}:
\begin{Lemm}\label{Projection4}
For each $\boldsymbol{X}=(X_d,X_f)\in \mathscr{C}$ and $\beta \ge 0$, we have
\begin{align}
R_0 ^{|\vLa |} e^{-\beta H_0} R_0 ^{|\vLa |} e^{-\beta H_0} R_0 ^{|\vLa |}
\unrhd \left( \frac{U^f}{2} \right)^{3|\vLa |} F_{\beta}({\bs X}) \wrt\p.
\end{align}
\end{Lemm}

\begin{proof}
 From \eqref{fnPP}, 
$n_{x, \up}^fn_{x, \down}^f\unrhd0 $ and $ (\mathbbm{1}-n_{x, \up}^f)(\mathbbm{1}-n_{x, \down}^f)\unrhd0\wrt\p $, so we know that 
\be
R_0  \unrhd \frac{U^f}{2}n_{x, \up}^fn_{x, \down}^f,\ R_0  \unrhd \frac{U^f}{2} (\mathbbm{1}-n_{x, \up}^f)(\mathbbm{1}-n_{x, \down}^f)\ \mbox{
w.r.t. $\p $}.
\ee
Thus, from \eqref{Def oPX-1}, the definitions of 
 $P_{\bs X}$ and $ Q_{\bs X}$, we have
\be
R_0 ^{|\vLa |} \unrhd  \left( \frac{U^f}{2} \right)^{|\vLa |}P_{\boldsymbol{X}},\quad R_0 ^{|\vLa |} \unrhd \left( \frac{U^f}{2} \right)^{|\vLa |}Q_{\boldsymbol{X}} \wrt\p.
\ee
Accordingly, by applying Lemmas \ref{InqSeki} and \ref{PPH_0}, we obtain the desired assertion.
\end{proof}

\begin{Lemm}\label{SemiOpLow}
For each $\bs X\in \mathscr{C}$ and $\beta\ge 0$, it holds that 
$
\{e^{-tH}\}_{t\geq0}\succeq F_\beta({\bs X}) \wrt\p 
$.
\end{Lemm}
\begin{proof}
Since $H=H_0-R$, one gets from Duhamel's formula that 
\be
e^{-tH}=\sum_{n=0}^{\infty} \mathscr{D}_n(t) \label{Duha1},
\ee
where $\mathscr{D}_0(t)=e^{-tH_0}$ and 
\begin{align}
\mathscr{D}_n(t) &= \int_{\Delta_n(t)} R(s_1) \cdots R(s_n) e^{-t H_0} d^n {\bs s},\quad R(s)=e^{-s H_0 }R e^{sH_0},\\
\Delta_n(t) &=\{{\bs s}=(s_1, \dots, s_n)\in \BbbR^n : 0\le s_1\le \cdots \le s_n \le t\}.
\end{align}
Note that the right-hand side of 
\eqref{Duha1} converges in the operator norm topology.
Since $R\unrhd 0 $ and $ e^{-sH_0} \unrhd 0$ w.r.t. $\p$ from 
Lemmas \ref{RProp} and \ref{PPH_0} respectively, we have
\be
K_{n, t}(\bs s)= R(s_1) \cdots R(s_n) e^{-t H_0} \unrhd 0\quad \mbox{w.r.t. $\p$}\quad ({\bs s} \in \Delta_n(t)).
\ee
Hence, by using Lemma \ref{Wcl}, we see that $\mathscr{D}_n(t) \unrhd 0$ w.r.t. $\p$. Therefore, we have
\be
e^{-tH} \unrhd \mathscr{D}_n(t)\quad \mbox{w.r.t. $\p$}\quad (n\in \{0\}\cup \BbbN).\label{LowSemi}
\ee

Now, assume that $\la \vphi|F_{\beta}(\bs X) \psi\ra>0$ holds for some $\vphi, \psi\in \p \setminus \{0\}$.
Choose $n=3|\vLa|$ and $ t=2\beta$. Define 
$\bs s_0\in \Delta_n(t)$ as 
\be
\bs s_0=\Big(\underbrace{0, \dots, 0}_{|\vLa|},
\underbrace{\beta, \dots, \beta}_{|\vLa|}, \underbrace{2\beta, \dots, 2\beta}_{|\vLa|}
\Big).
\ee
Then, we readily confirm that 
\be
K_{n, t}(\bs s_0)=R^{|\vLa|} e^{-\beta H_0} R^{|\vLa|} e^{-\beta H_0} R^{|\vLa|}. \label{KExp}
\ee
Since $R_i\unrhd 0$ w.r.t. $\p$ ($i=0, 1$) from  Lemma \ref{RProp}, we get $R\unrhd R_0$ w.r.t. $\p$. Hence, by using Lemmas \ref{InqSeki}, \ref{PPH_0} and \ref{Projection4}, one finds that 
\be
K_{n, t}(\bs s_0)\unrhd R_0^{|\vLa|} e^{-\beta H_0} R_0^{|\vLa|} e^{-\beta H_0} R_0^{|\vLa|} \unrhd \left( \frac{U^f}{2} \right)^{3|\vLa |} F_{\beta}({\bs X}) \wrt\p,
\ee
which implies that 
\be
\la \vphi|K_{n, t}(\bs s_0)\psi\ra \ge \left( \frac{U^f}{2} \right)^{3|\vLa |} \la \vphi|F_{\beta}({\bs X})\psi\ra>0.
\ee
Because $K_{n, t}(\bs s)$ is strongly continuous in  $\bs s$, we obtain $\la \vphi|\mathscr{D}_n(t) \psi\ra>0$. Combining this with \eqref{LowSemi} we conclude that $\la \vphi|e^{-tH} \psi\ra>0$. This completes the proof of Lemma \ref{SemiOpLow}.
\end{proof}

\subsubsection*{Completion of the proof of Proposition \ref{BP2}}
From Lemma \ref{MonoE}, $e^{-\beta H}\unrhd e^{-\beta H_0}\wrt\p \ (\beta \ge 0)$ holds.
Furthermore, according to  Lemma \ref{SemiOpLow},  it follows that $\{e^{-tH}\}_{t\geq0}\succeq F_\beta({\bs X}) \wrt\p$ for each $\bs X\in \mathscr{C}$ and $\beta \ge 0$.
Hence, by using Lemma \ref{succ prod} and Remark \ref{Yowai}, we can conclude that $\{e^{-tH}\}_{t\geq0}\succeq F_{\beta, \beta'}(\bs p) \wrt\p$ holds. \qed

\section{Proof of Proposition \ref{BP3}}\label{Sec8}

\subsection{Outline of the proof}

To prove Proposition \ref{BP3}, we prepare three propositions.

To state the first proposition, we introduce some symbols: for each 
$\bs X=(X_f, X_d)\in \mathscr{C}$, define
\be
\mathbb{E}_{\bs X}=\prod_{\sigma=\up, \down} P^d_{X_d\setminus X_f, \sigma} \overline{P}^d_{X_f\setminus X_d, \sigma} P_{X_f, \sigma}^f \overline{P}^f_{\vLa\setminus X_f, \sigma}, 
\ee
where the operators $P^f_{X, \sigma}, P^d_{X, \sigma}$, etc. are defined in \eqref{DefPPPP2}. 
$\mathbb{E}_{\bs X}$ is an orthogonal projection and plays an essential role in the following discussion.

\begin{Prop}\label{Projection3} For every $\bs X\in \mathscr{C}$ and $\beta>0$, it holds that 
\begin{align}
F_{\beta}(\bs X)
&=\beta^{4|\bs X|_{\triangle}} V^{4|\bs X|_{\triangle} }\Ex_{\boldsymbol{X}}
+o(\beta^{4|\bs X|_{\triangle}}),
\end{align}
where  $o(\beta^n)$ is some bounded operator satisfying 
\be
\lim_{\beta\to +0}\frac{o(\beta^n)}{\beta^n} \vphi=0\quad (\vphi\in \h).
\ee
\end{Prop}
The proof of Proposition \ref{Projection3} is given in Appendix \ref{PfProjection3}.

For each path $\bs p=\bs X_1\cdots\bs X_n$, define
\be
\Ex_{\beta}(\bs p)=
\mathbb{E}_{\bs X_1}e^{-\beta H_0} \mathbb{E}_{\bs X_2} e^{-\beta H_0} \cdots \mathbb{E}_{\bs X_{n-1}} e^{-\beta H_0} \mathbb{E}_{\bs X_n}.
\ee

\begin{Prop}\label{connectivity4}
For any given $\boldsymbol{X}\in \mathscr{C}$, consider a path $\bs p=\bs X_0\bs X_1\cdots\bs X_{n+1}\ (\bs X_0=\bs F,\ \bs X_{n+1}=\bs X)$ connecting $\bs X$ and $\bs F$.
Then, one obtains 
\begin{align}
\Ex_{\beta}(\bs p)
=\beta^{2n+2}J_{\beta}(\bs p)D({\bs p})+o(\beta^{2n+2}),  \label{connectivity4-1}
\end{align}
where $D({\bs p})$ is a bounded operator on $\h_{\rm e}$ satisfying 
\begin{align}
D({\bs p}) \ket{{\bs X}, {\bs X}}
=c \ket{{\bs F}, {\bs F}}\quad(c>0),
\end{align}
 and $J_{\beta}(\bs p)$ is a bounded operator on $L^2(\BbbR^{|\vLa|})$satisfying
\be
 \lim_{\beta\to +0} J_{\beta}(\bs p)=\mathbbm{1} \quad \mbox{in the strong operator topology.}\label{ConTo1}
\ee
\end{Prop}
We prove Proposition \ref{connectivity4} in Subsection \ref{PfCo4}.

\begin{Prop}\label{AEE}
For each 
$\varphi\geq0\wrt \p $ with $\vphi\neq 0$, there exists a $\boldsymbol{Z}\in \mathscr{C}$ such that $\varphi_{\bs Z, \bs Z}\neq 0$, 
where $\vphi_{\bs Z, \bs Z}$ is defined by \eqref{vphiDec}.
Let $\boldsymbol{X}$ be one of such $\boldsymbol{Z}$ with the largest $|{\bs Z}|_{\triangle}$.
Then, we have 
\begin{align}
\Ex_{\boldsymbol{X}}\varphi
=|\boldsymbol{X},\boldsymbol{X} \rangle\otimes  \vphi_{\bs X, \bs X}. \label{AEE-1}
\end{align}

\end{Prop}
The proof of Proposition \ref{AEE} is given in Subsection \ref{PfAEE}.
Given the above three propositions, we can prove Proposition \ref{BP3}:

\subsubsection*{Proof of Proposition \ref{BP3} given Propositions \ref{Projection3}, \ref{connectivity4} and  \ref{AEE}}

From Proposition  \ref{AEE}, there exists an $\bs X\in \mathscr{C}$ such that $\Ex_{\bs X} \vphi= \ket{\bs X, \bs X} \otimes \vphi_{\bs X, \bs X}$, $\vphi_{\bs X, \bs X}\ge 0$ and  $\vphi_{\bs X, \bs X} \neq 0$.
Then take a path $\bs p=\bs X_0\bs X_1\cdots \bs X_{n+1}$  (${\bs X}_0=\bs F$ and $\bs X_{n+1}=\bs X$) connecting 
$\bs X$ and $\bs F$.
From Proposition \ref{connectivity4}, it follows that 
\begin{align}
\Ex_{\beta}({\bs p})\vphi
=c_0\beta^{2n+2} | {\bs F}, {\bs F}; f_{\beta} \rangle +o(\beta^{2n+2})\vphi  \quad (c_0>0), \label{connectivity2-1}
\end{align}
where $f_{\beta}=J_{\beta}(\bs p) \vphi_{\bs X, \bs X}$.

On the other hand,  by using Proposition  \ref{Projection3}, we have
\begin{align}
\lim_{\beta'\to+0}\beta'^{-4|\bs X|_\triangle}F_{\beta'}(\boldsymbol{X})\vphi
= V^{4|\boldsymbol{X}|_\triangle}\Ex_{\boldsymbol{X}}\vphi.\label{FLim}
\end{align}
Combining 
\eqref{connectivity2-1} with \eqref{FLim}, we find that 
\begin{align}
&\lim_{\beta\rq{}\to +0}\beta'^{-4d} F_{\beta'}({\bs X}_0) e^{-\beta H_0} \cdots e^{-\beta H_0} F_{\beta'}({\bs X}_n)\vphi \no
&=c_0\beta^{2n+2}\Bigg(\prod_{i=0}^n V^{4|{\bs X}_i|_\triangle}\Bigg) | {\bs F}, {\bs F}; f_{\beta} \rangle
+o(\beta^{2n+2}) \vphi.
\end{align}
If we set $f=\vphi_{\bs X, \bs X}$, then from \eqref{ConTo1}, $f_{\beta} \to f\ (\beta\to +0)$ holds.
Therefore, by choosing 
$
c
=c_0\big(\prod_{i=0}^nV^{4|{\bs X}_i|_\triangle}\big)
$, 
we finally arrive at 
\begin{align}
\lim_{\beta\to+0} \lim_{\beta'\to+0}
\beta^{-2n-2}\beta'^{-4d} F_{\beta, \beta\rq{}}(\bs p)\vphi
=c | {\bs F}, {\bs F}; f \rangle.
\end{align}
\qed

\subsection{Proof of Proposition \ref{connectivity4}}\label{PfCo4}
\subsubsection{Characteristics of the edges}

Given $x,  y\in \vLa $, we set
\begin{align}
\hat{B}_{x,  y}^-=
\begin{cases}
t_{x, y}\hat{d}_x^*\hat{d}_y\ \ &(x\neq y)\\
V\hat{f}_x^*\hat{d}_x\ \ &(x=y)
\end{cases},\quad 
\hat{B}_{x,  y}^+=
\begin{cases}
t_{x, y}\hat{d}_y^*\hat{d}_x\ \ &(x\neq y)\\
V\hat{d}_x^*\hat{f}_x\ \ &(x=y)
\end{cases}
\label{DefB}
\end{align}
and 
\be
\hat{B}_{x, y}=\hat{B}_{x, y}^-+\hat{B}_{x, y}^+.
\ee

The following facts are often used in the proof of Proposition \ref{connectivity4}:
\begin{Prop}\label{EquivC}
Let $\bs X, \bs Y\in \mathscr{C}$. The following {\rm (i), (ii)}, and {\rm (iii)} are equivalent to each other:

\begin{itemize}
\item[\rm (i)] $\{{\bs X} ,{\bs Y}\}$ is an edge.
\item[\rm (ii)] 
There exist $x, y\in \vLa$ and non-zero constant $c$ such that $\hat{B}_{x, y}|{\bs X}\ra=c \ket{\bs Y}$ holds.
\item[\rm (iii)] There exist $x, y\in \vLa$ and non-zero constant $c$ such that either $\hat{B}_{x, y}^-|{\bs X}\ra=c \ket{\bs Y}$ or $\hat{B}_{x, y}^+|{\bs X}\ra=c \ket{\bs Y}$ holds.
\end{itemize}

\end{Prop}
\begin{proof}

For each $\bs X=(X_d, X_f)\in \mathscr{C}$, we denote by $\ket{X_d; X_f}$ 
the $\ket{\bs X}$ defined by \eqref{DefVecX}; this notation has the advantage that the electron configurations of the $d$-and $f$-electrons becomes clearer.
For each $\bs W=(W_d, W_f)\in \mathscr{C}$ and $u\in \vLa$, it holds that 
\begin{align}
\hat{B}^-_{u, u} \Big| W_d;  W_f\Big\ra
&=
\pm V \Big| W_d \setminus \{u\}; W_f\cup \{u\}\Big\ra&  \mbox{if $u\notin W_f,\ u\in W_d$},\label{Bket}\\ 
\hat{B}^+_{u, u} \Big| W_d;  W_f\Big\ra
&=
\pm V \Big| W_d \cup\{u\}; W_f\setminus \{u\}\Big\ra&  \mbox{if $u\in W_f,\ u\notin W_d$}, \label{BketB}\\ 
\hat{B}^{\pm}_{u, u} \Big| W_d;  W_f\Big\ra
&=0&  \mbox{otherwise}, 
\end{align}
where $\pm A$ represents either $A$ or $-A$.\footnote{In the following discussion, it is not necessary to explicitly determine whether $+$ or $-$, so we dare to leave it ambiguous in this way.} 
Similarly, 
for each $\bs W=(W_d, W_f)\in \mathscr{C}$ and $u, v\in \vLa\  (u\neq v)$,  the following hold:
 \begin{align}
 \hat{B}_{u, v }^- \Big| W_d; W_f\Big\ra &=
 \pm \tilde{t}_{u, v} \Big|
 (W_d\cup \{u\})\setminus \{v\} ; W_f \Big\ra & \mbox{if $u\notin W_d,\ v\in W_d$,} \label{Bket1}\\
  \hat{B}_{u, v }^+ \Big| W_d; W_f\Big\ra&=
 \pm \tilde{t}_{u, v} \Big|
 (W_d\cup \{v\})\setminus \{u\} ; W_f \Big\ra & \mbox{if $u\in W_d,\ v\notin W_d$,}\label{Bket2}\\
  \hat{B}_{u, v }^{\pm} \Big| W_d; W_f\Big\ra&=
 0 & \mbox{otherwise}\label{Bket3}.
 \end{align}
From these equations and Definition \ref{DefNei}, the claim of the proposition follows immediately.
\end{proof}

\begin{Def} \upshape
For each 
${\bs X}=(X_{d},  X_f)\in \mathscr{C}$, we define
\be
E_{\bs X}=
\Bigg[
\prod_{x\in X_d\setminus X_f} \hat{n}_x^d
\Bigg]
\Bigg[
\prod_{x\in X_f\setminus X_d}\overline{\hat{n}}^d_x
\Bigg]
\Bigg[
\prod_{x\in X_f} \hnf_x
\Bigg]
\Bigg[
\prod_{x\in \vLa\setminus X_f} \overline{\hat{n}}_x^f
\Bigg],
\ee
where
\be
\hnd_x=\hd_x^*\hd_x,\ \hnf_x=\hf_x^*\hf_x,\ \overline{\hat{n}}_x^d=\mathbbm{1}-\hnd_x,\ \overline{\hat{n}}_x^f
=\mathbbm{1}-\hnf_x.
\ee
$E_{\bs X}$ is an orthogonal projection on $\mathfrak{E}$.
Note that under the identification \eqref{IdnE}, we can express $\mathbb{E}_{\bs X}$ as  $\mathbb{E}_{\bs X}=E_{\bs X}\otimes E_{\bs X}$.
\end{Def}

\begin{Prop}\label{PropNei}
Let 
$\bs X, \bs Y\in \mathscr{C}$.
Suupose that $\{\bs X, \bs Y\}$ is an edge such that $X_d\triangle Y_d=\{x, y\}$,
where if 
$\{\bs X, \bs Y\}$ is a $(d, f)$-edge, then we understand $x=y$ and $\{x, y\}=\{x\}$.
Then the following hold for either $\vepsilon=+$ or $-$: 
\begin{align}
E_{\bs X} E_{\bs Y}&=0; \label{E1}\\
\hat{B}^{\vepsilon}_{x, y}\ket{\bs X}&=c\ket{\bs Y}\quad(c\neq 0); \label{E2} \\
\hat{B}^{\vepsilon}_{x, y}E_{\bs X}&=\pm E_{\bs Y} \hat{B}_{x, y}^{\vepsilon},\quad   E_{\bs X} \hat{B}_{x, y}^{\overline{\vepsilon}}=\pm \hat{B}^{\overline{\vepsilon}}_{x, y}E_{\bs Y}; \label{E3}\\
 E_{\bs X} \hat{H}_1(\VP(\bs q)) E_{\bs Y}&=\pm E_{\bs X} e^{\pm \im \vP_{\{x, y\}}}\hat{B}^{\overline{\vepsilon}}_{x, y}, \label{E4}
\end{align}
where $\overline{\vepsilon}=1$ if $\vepsilon=0$, $\overline{\vepsilon}=0$ if $\vepsilon=1$; $\pm A$ represents either $A$ or $-A$;\footnote{Same as the above footnote.}  $\Phi_{\{x, y\}}=\Phi_{x, y}$ if $x\neq y$, $\Phi_{\{x, y\}}=\Phi_x$ if $x=y$.
Here, recall that $ \hat{H}_1(\VP(\bs q)) $ is defined by \eqref{DefhH_1}.
\end{Prop}

\begin{proof}
{\it Case 1.} The case where $\{\bs X, \bs Y\}$ is a $d$-edge.
First, we note that $E_{\bs X}$ can be expressed as 
\be
E_{\bs X}=\sum_{{Z\subseteq \vLa \setminus (X_d \triangle X_f)}\atop{|Z|=|X_d\cap X_f|}} \Big| (X_d\setminus X_f)\cup Z; X_f\Big\ra\Big \la(X_d\setminus X_f)\cup Z; X_f \Big|.  \label{BaseE}
\ee

(i) We divide the proof into two cases:
 \begin{itemize}
 \item[] (i-a) $Y_d=X_d\cup \{x\},\ Y_f=X_f \setminus \{x\}$.
 \item[] (i-b) $Y_d=X_d\setminus \{x\},\ Y_f=X_f\cup \{x\}$.
 \end{itemize}
In the case of (i-a), we readily confirm 
\eqref{E1} and \eqref{E2}. Below, we show \eqref{E3} and \eqref{E4}.
By using \eqref{Bket} and \eqref{BketB}, one has 
\begin{align}
\hat{B}_{x, x}^+ \Big| (X_d\setminus X_f)\cup Z; X_f\Big\ra&= \pm V \Big| (X_d\setminus X_f)\cup Z \cup\{x\}; X_f\setminus \{x\}\Big\ra\no
&=\pm V \Big| (Y_d\setminus Y_f)\cup Z; Y_f\Big\ra,\\
\hat{B}_{x, x}^- \Big| (Y_d\setminus Y_f)\cup Z'; Y_f\Big\ra&=\pm V \Big| \big( (Y_d\setminus Y_f)\cup Z' \big) \setminus \{x\}; Y_f\cup \{x\}\Big\ra\no
&=\pm V \Big| (X_d\setminus X_f)\cup Z'; X_f\Big\ra.
\end{align}
Combining these equations with \eqref{BaseE}, we obtain \eqref{E3}.
The idea of the proof of \eqref{E4} is as follows: first, we rewrite $E_{\bs X}\hat{H}_1(\VP)E_{\bs Y}$ as
\begin{align}
E_{\bs X}\hat{H}_1(\VP)E_{\bs Y}= -E_{\bs X} \{ e^{-\im \vP_x}\hat{B}^-_{x, x} 
+e^{\im \vP_x}\hat{B}^+_{x, x} 
\}E_{\bs Y}+E_{\bs X}I_x E_{\bs Y}.
\end{align}
Using \eqref{E3}, the first term on the right-hand side is equal to $\pm e^{\overline{\vepsilon} \im \vP_x} E_{\bs X} \hat{B}_{x, x}^{\overline{\vepsilon}}$. The second term on the right-hand side contains, for example, the following term:
$
\sum_{z\neq x} E_{\bs X}\hat{B}_{z, z}E_{\bs Y}
$. If $z\neq x$, we know that $E_{\bs X} \hat{B}_{z, z} E_{\bs Y}=0$.
In this way, one obtains $E_{\bs X}I_x E_{\bs Y}=0$. In the case of (i-b), we can similarly show that \eqref{E1}-\eqref{E4}.
\medskip

{\it Case 2.} The case where $\{\bs X, \bs Y\}$ is a $(d, f)$-edge.
It is easy to see that \eqref{E1} and \eqref{E2} are valid. Below we describe the strategy for the proof of \eqref{E3} and \eqref{E4}.
The proof is divided into the following three cases:
 \begin{itemize}
 \item[] (ii-a) $x, y\in X_d\setminus X_f$.
 \item[] (ii-b) $x\in X_d\setminus X_f,\ y\in X_d\cap X_f$.
 \item[] (ii-c) $x, y\in X_d \cap X_f$.
 \end{itemize}
 In each case, the proof is given in a manner similar to that of (i). For the readers' convenience, we describe the proof strategy for  the  case (ii-a) below.
 The same applies to (ii-b) and (ii-c).
 
  (ii-a) can be further divided into two cases:
 \begin{itemize}
 \item[] (ii-a-1) $Y_d=(X_d\cup \{y\}) \setminus \{x\}$.
 \item[] (ii-a-2) $Y_d=(X_d\cup \{x\}) \setminus \{y\}$.
 \end{itemize}
 In case (ii-a-1), since $Y_d\setminus Y_f=\big((X_d\setminus X_f) \cup \{y\} \big)\setminus \{x\}$,
 we find by using \eqref{Bket1}  and \eqref{Bket2} that 
 \begin{align}
 \hat{B}_{x, y}^{+} \Big| (X_d\setminus X_f) \cup Z; X_f\Big\ra &=\pm t_{x, y} \Big|
 (Y_d\setminus Y_f) \cup Z; Y_f \Big\ra,\\
 \hat{B}_{x, y}^{-} \Big| (Y_d\setminus Y_f) \cup Z'; Y_f\Big\ra &=\pm t_{x, y} \Big|
 (X_d\setminus X_f) \cup Z'; X_f \Big\ra.
 \end{align}
 Combining these equations with \eqref{BaseE}, we obtain \eqref{E3}.
 The idea of the proof of \eqref{E4} is as follows: 
 first, we express $E_{\bs X}\hat{H}_1(\VP)E_{\bs Y}$ as
 \begin{align}
E_{\bs X}\hat{H}_1(\VP)E_{\bs Y}= -E_{\bs X} \{ e^{\im \vP_{x, y}}\hat{B}^-_{x, y} 
+e^{-\im \vP_{x, y}}\hat{B}^+_{x, y} 
\}E_{\bs Y}+E_{\bs X}I_{x, y} E_{\bs Y}.
\end{align}
By using \eqref{E3}, the first term on the right-hand side is equal to $\pm e^{\im \vepsilon \vP_{x, y}}E_{\bs X} \hat{B}_{x, y}^{\overline{\vepsilon}}$.
For the second term on the right-hand side, using \eqref{Bket3}, we can show that $E_{\bs X}I_{x, y} E_{\bs Y}=0$.
 The case (ii-a-2) can be proved in a similar way.
\end{proof}

\subsubsection{Completion of the proof of Proposition \ref{connectivity4}}

For each $\sigma=\up, \down$ and $x, y\in \vLa$, we set
\begin{align}
B_{x,  y; \sigma}^-=
\begin{cases}
t_{x, y}d_{x, \sigma}^*d_{y, \sigma}\ \ &(x\neq y)\\
Vf_{x, \sigma}^* d_{x, \sigma}\ \ &(x=y)
\end{cases},\quad 
B_{x,  y; \sigma}^+=
\begin{cases}
t_{x, y} d_{y, \sigma}^* d_{x, \sigma}\ \ &(x\neq y)\\
V d_{x, \sigma}^*f_{x, \sigma}\ \ &(x=y).
\end{cases}
\end{align}
Under the identification \eqref{IdnE}, we have $B^{\vepsilon}_{x, y; \up}=\hat{B}^{\vepsilon}_{x, y} \otimes \mathbbm{1}$ and 
$B^{\vepsilon}_{x, y; \down}=\mathbbm{1} \otimes \hat{B}^{\vepsilon}_{x, y}$ for $\vepsilon =+, -$, where $\hat{B}_{x, y}^{\vepsilon}$ is defined by \eqref{DefB}.
 We also set 
 \be
B_{x, y}^{\vepsilon}=B_{x, y; \up}^{\vepsilon} B_{x, y; \down}^{\vepsilon}\quad (\vepsilon=+, -).
\ee

For given $x, y\in \vLa$, we define 
$J_{x, y}(\beta)$ as follows: 
\begin{description}
\item[]
  If $x=y$, then
\begin{align}
 J_{x, x}(\beta)=\int_{0\le s_1 \le s_2\le 1}
& \Big(
 e^{-s_1 \beta \omega_0\Np} e^{-\im \vPhi_x} e^{-(s_2-s_1) \beta \omega_0\Np} e^{+\im \vPhi_x} e^{-(1-s_2)\beta \omega_0 \Np}\no
 &+e^{-s_1 \beta \omega_0 \Np} e^{+\im \vPhi_x} e^{-(s_2-s_1)\beta \omega_0  \Np} e^{-\im \vPhi_x} e^{-(1-s_2) \beta \omega_0\Np}
 \Big)ds_1ds_2;
 \end{align}
 \item[] If $x\neq y$, then
\begin{align}
 J_{x, y}(\beta)=\int_{0\le s_1 \le s_2\le 1 }
& \Big(
 e^{-s_1 \beta \omega_0 \Np} e^{-\im \vPhi_{x, y}} e^{-(s_2-s_1) \beta \omega_0 \Np} e^{+\im \vPhi_{x, y}} e^{-(1-s_2) \beta \omega_0 \Np}\no
 &+e^{-s_1 \beta \omega_0  \Np} e^{+\im \vPhi_{x, y}} e^{-(s_2-s_1) \beta \omega_0 \Np} e^{-\im \vPhi_{x, y}} e^{-(1-s_2) \beta \omega_0 \Np}
 \Big)ds_1ds_2.
 \end{align}
 
 \end{description}
 
 \begin{Lemm}\label{EBase}
 Let $\bs X, \bs Y\in \mathscr{C}$.
 Suppose that $\{\bs X, \bs Y\}$ is an edge with $X_d\triangle Y_d=\{x, y\}$, where if 
$\{\bs X, \bs Y\}$ is a $(d, f)$-edge, then we understand $x=y$ and $\{x, y\}=\{x\}$. Then, for either $\vepsilon=-$ or $+$, the following holds:
 \be
 \Ex_{\bs X} e^{-\beta H_0} \Ex_{\bs Y}= \beta ^2 J_{x, y}(\beta) \Ex_{\bs X} B_{x, y}^{\vepsilon}+O(\beta^3),
 \ee
 where $O(\beta^n)$ is some bounded operator satisfying
\be
\limsup_{\beta\to +0}\frac{\|O(\beta^n)\|}{\beta^n}<\infty.
\ee

 \end{Lemm}
 \begin{proof}
 For each $x\in \vLa$, define
  \begin{align}
 \hat{B}_{x,x}^-(\VP)=e^{-\im \vPhi_x}\hat{B}_{x,x}^-, \quad \hat{B}_{x, x}^+(\VP)=e^{+\im \vPhi_x} \hat{B}_{x, x}^+.
 \end{align}
 Under the identification \eqref{IdnE}, we obtain 
 \begin{align}
 V_{\up}(\VP) &=\sum_{x\in \vLa}\Big\{\hat{B}^-_{x, x}(\VP)+ \hat{B}^+_{x, x}(\VP)\Big\}\otimes \mathbbm{1}, \\
 V_{\down}(-\VP)&=\sum_{x\in \vLa}\mathbbm{1} \otimes \Big\{\hat{B}^-_{x, x}(-\VP)+ \hat{B}^+_{x, x}(-\VP)\Big\}.
 \end{align}
 Similarly, if we set, for each $x, y\in \vLa\ (x\neq y)$, 
 \be
 \hat{B}_{x,y}^-(\VP)=e^{+\im \vPhi_{x, y}}\hat{B}_{x,y}^-, \quad \hat{B}_{x, y}^+(\VP)=e^{-\im \vPhi_{x, y}} \hat{B}_{x, y}^+, 
 \ee
 then we have
 \begin{align}
 T_{\up}(\VP)&=\sum_{\{x, y\}\in E} t_{x, y}\Big\{ \hat{B}_{x, y}^-(\VP)+\hat{B}^+_{x, y}(\VP)\Big\} \otimes \mathbbm{1},\\
 T_{\down}(-\VP)&=\sum_{\{x, y\}\in E} t_{x, y}\mathbbm{1} \otimes \Big\{\hat{B}_{x, y}^-(-\VP)+\hat{B}^+_{x, y}(-\VP)\Big\}, 
 \end{align}
 where $E$ is the set of edges defined immediately above the assumption \hyperlink{A2}{\bf (A. 2)}.
 Recall that 
 \be
 \mathbb{E}_{\bs X}=E_{\bs X} \otimes E_{\bs X}. \label{TensorEE}
 \ee

 From Duhamel\rq{}s formula, it follows that 
\begin{align}
&\Ex_{\bs X} e^{-\beta H_0}\Ex_{\bs Y}\no
=&\Ex_{\bs X} \bigg\{\mathbbm{1}+\beta \int_0^{1} H_1(s) e^{- \beta \omega_0 \Np}ds
+\beta^2 \int_{0\le s_1\le s_2\le 1} H_1(s_1) H_1(s_2) e^{-  \beta \omega_0 \Np}
\bigg\}\Ex_{\bs Y}+O(\beta^3),
\end{align}
where, for a give operator $A$, we set
$A(s)=e^{-s \beta \omega_0 \Np} A e^{s \beta \omega_0 \Np}$.
Using 
Proposition \ref{PropNei} and \eqref{TensorEE}, we find that 
\be
\Ex_{\bs X} \bigg\{\mathbbm{1}+\beta \int_0^{1} H_1(s) e^{-\beta \omega_0 \Np}ds\bigg\}\Ex_{\bs Y}=0.
\ee

First consider the case where $\{\bs X, \bs Y\}$ is a $(d, f)$-edge.
Again from   Proposition \ref{PropNei} and \eqref{TensorEE}, we see  that  for some $\vepsilon\in \{-, +\}$, 
the following holds:
\begin{align}
&\int_{0\le s_1\le s_2 \le 1 }\Ex_{\bs X}  H_1(s_1) H_1(s_2) \Ex_{\bs Y} e^{-\beta \omega_0 \Np} ds_1ds_2\no
=&
\int_{0\le s_1\le s_2 \le 1 }\Ex_{\bs X} \Big\{
V_{\up}(\VP)(s_1) V_{\down}(-\VP)(s_2)+V_{\down}(-\VP)(s_1) V_{\up}(\VP)(s_2)
\Big\}\Ex_{\bs Y} e^{-\beta \omega_0 \Np} ds_1ds_2\no
=&
J_{x, x}(\beta) B^{\vepsilon}_{x, x}.
\end{align}
Consequently, the claim holds when $\{\bs X, \bs Y\}$ is a $(d, f)$-edge.

Next consider the case where $\{\bs X, \bs Y\}$ is a $d$-edge.
From Proposition \ref{PropNei} and \eqref{TensorEE}, it follows  that  for some $\vepsilon\in \{-, +\}$, the following holds:
\begin{align}
&\int_{0\le s_1\le s_2 \le 1}\Ex_{\bs X}  H_1(s_1) H_1(s_2) \Ex_{\bs Y} e^{-\beta \omega_0 \Np} ds_1ds_2\no
=&
\int_{0\le s_1\le s_2 \le 1}\Ex_{\bs X} \Big\{
T_{\up}(\VP)(s_1) T_{\down}(-\VP)(s_2)+T_{\down}(-\VP)(s_1) T_{\up}(\VP)(s_2)
\Big\}\Ex_{\bs Y} e^{-\beta \omega_0 \Np} ds_1ds_2\no
=&
J_{x, y}(\beta) B^{\vepsilon}_{x, y}.
\end{align}
Thus, the claim is also shown when $\{\bs X, \bs Y\}$ is a $d$-edge.
 \end{proof}

 Under the above preparation, we can easily prove Proposition \ref{connectivity4}: 
Using Lemma \ref{EBase} and \eqref{E3} repeatedly, we see that there exists some $\bs \vepsilon=(\vepsilon_j)_{j=0}^{n}\in \{-, +\}^{n+1}$ such that the following holds:
\begin{align}
\Ex_{\beta}(\bs p)
= \beta^{2n+2}  J_{\beta} (\bs p) \Ex_{\boldsymbol{F}} B^{\bs \vepsilon}(\boldsymbol{p})
+ O(\beta^{2n+3}),  \label{E(p)}
\end{align}
where 
\begin{align}
B^{\bs \vepsilon}(\boldsymbol{p})
=B^{\vepsilon_0}_{x_0, y_0} B^{\vepsilon_1}_{x_1, y_1}\cdots B^{\vepsilon_n}_{x_n, y_n}.
\end{align}
Then, setting $D(\bs p)=\mathbb{E}_{\bs F} B^{\bs \vepsilon}(\bs p)$, we find 
\be
D(\bs p) \ket{\bs X, \bs X}=c\ket{\bs F, \bs F}\quad(c> 0)
\ee
by repeatedly using \eqref{E2}. 
It is clear from the definition of $J_{\beta}(\bs p)$ that \eqref{ConTo1} is true.
This completes the proof of Proposition \ref{connectivity4}.
\qed

\subsection{Proof of Proposition \ref{AEE}}\label{PfAEE}

In this subsection, we will complete the proof of Proposition \ref{BP3} by giving the proof of Proposition \ref{AEE}. First, we prepare a lemma:

\begin{Lemm} \label{phi max}
For each 
$\varphi\geq0\wrt \p $ with $\vphi\neq 0$, there exists a $\boldsymbol{Z}\in \mathscr{C}$ such that $\varphi_{\bs Z, \bs Z} \neq 0$.
Let $\boldsymbol{X}$ be one of such $\boldsymbol{Z}$ with the largest $|{\bs Z}|_{\triangle}$.
Let $\boldsymbol{X}'$ be an electron configuration satisfying $E_{\boldsymbol{X}}|\boldsymbol{X}'\rangle=|\boldsymbol{X}'\rangle$ and $\vphi_{\boldsymbol{X}',\boldsymbol{X}'}\neq 0$. Then 
$X_d\setminus X_f=X'_d\setminus X'_f$ and $X_f\setminus X_d=X'_f\setminus X'_d$ hold.

\end{Lemm}

\begin{proof}
We have already shown that there exists a $\boldsymbol{Z}\in \mathscr{C} $ satisfying $\vphi_{\bs Z, \bs Z} \neq 0$ in the discussion above Proposition \ref{BP3}.

Assume that 
$X_d\setminus X_f\neq X'_d\setminus X'_f$ or $X_f\setminus X_d\neq X'_f\setminus X'_d$.
From the maximality of $|\boldsymbol{X}|_{\triangle}$, we have
\be
(X_d\setminus X_f)\setminus(X'_d\setminus X'_f)\neq \varnothing \ \ \mbox{or}\ \ (X_f\setminus X_d)\setminus(X'_f\setminus X'_d)\neq\varnothing.
\ee
Consider  the case where
$(X_d\setminus X_f)\setminus(X'_d\setminus X'_f)\neq \varnothing$.  If  we take $x$ as  $x\in(X_d\setminus X_f)\setminus(X'_d\setminus X'_f)$, then $E_{\bs X}$ contains  $\hnd_x \ohnf_x$.
Hence, since  $x\not\in X'_d\setminus X'_f$,  $E_{\boldsymbol{X}}|\boldsymbol{X}'\rangle=0$ holds. 
A similar argument shows that $E_{\boldsymbol{X}}|\boldsymbol{X}'\rangle=0$ in the case where $(X_f\setminus X_d)\setminus(X'_f\setminus X'_d)\neq\varnothing$. Since the consequences in both cases contradict with $E_{\boldsymbol{X}}|\boldsymbol{X}'\rangle=|\boldsymbol{X}'\rangle$, we can conclude that $X_d\setminus X_f=X'_d\setminus X'_f$ and $X_f\setminus X_d=X'_f\setminus X'_d$ hold.
\end{proof}

\subsubsection*{Completion of the proof of  Proposition \ref{AEE}}
Let $\bs X\in \mathscr{C}$ be the one in Proposition \ref{AEE}.
If we express $\varphi$ as \eqref{vphiDec}, we obtain 
\begin{align}
\Ex_{\boldsymbol{X}}\varphi
=\underset{\boldsymbol{X}',\boldsymbol{X}''}{{\sum}'} |\boldsymbol{X}',\boldsymbol{X}''\rangle
\otimes \varphi_{\boldsymbol{X}',\boldsymbol{X}''},
\end{align}
 where 
 the right-hand side indicates that the sum is taken over 
 $|\boldsymbol{X}',\boldsymbol{X}''\rangle
\otimes \varphi_{\boldsymbol{X}',\boldsymbol{X}''}$
  such that \be E_{\boldsymbol{X}}|\boldsymbol{X}'\rangle=|\boldsymbol{X}'\rangle,\quad E_{\boldsymbol{X}}|\boldsymbol{X}''\rangle=|\boldsymbol{X}''\rangle.\label{RistXX}
  \ee
From Lemma \ref{phi max}, for  $\boldsymbol{X}'$ and ${\bs X}''$
satisfying \eqref{RistXX}, we have
 \be
X_f\setminus X_d=X_f'\setminus X_d'=X_f''\setminus X_d'',\quad X_d\setminus X_f=X_d'\setminus X_f'=X_d''\setminus X_f''.
\ee 
In addition, $\Ex_{\boldsymbol{X}} |\boldsymbol{X}', \boldsymbol{X}''\rangle=(E_{\bs X}\ket{\bs X\rq{}})\otimes(E_{\bs X}\ket{\bs X\rq{}\rq{}} )\neq0$ holds only if  $X_f=X_f'=X_f''$. 
Therefore, we obtain $X_d'\cap X_f'= X_f'\setminus(X_f'\setminus X_d')= X_f\setminus(X_f\setminus X_d)= X_d\cap X_f$. Similarly, we see that $X_d''\cap X_f''=X_d\cap X_f$ holds. 
From the above,  $\Ex_{\boldsymbol{X}} |\boldsymbol{X}',\boldsymbol{X}''\rangle\neq0$ is valid only if  $\boldsymbol{X}=\boldsymbol{X}'=\boldsymbol{X}''$.
Hence, (\ref{AEE-1}) holds. \qed

\section{Completion of the proof of the main theorems}\label{Sec9}

\subsection{Total spin of the ground state of $H$}

This subsection proves the following proposition:
\begin{Prop}\label{GTotalSP}
The ground state of $H$ has total spin $S=0$.
\end{Prop}
We employ the overlap principle developed by the authors in \cite{MIYAO2021168467} to prove Proposition \ref{GTotalSP}.
For this purpose, we introduce the following Hamiltonian:
\begin{align}
H_{\mathrm{H}}
=\sum_{x,y\in\vLa}\sum_{\sigma=\up, \down}t_{x,y}d_{x, \sigma}^*d_{y, \sigma}
+\sum_{x\in\vLa}\sum_{\sigma=\up, \down}(f_{x, \sigma}^*d_{x, \sigma}+ d_{x, \sigma}^*f_{x, \sigma})
+\sum_{x\in\vLa}(n_{x, \up}^dn_{x, \down}^d +n_{x, \up}^fn_{x, \down}^f)
\end{align}
Note that $H_{\rm H}$ acts in  $\h_{\rm e}$.

\begin{Lemm}\label{HubbardPI}
For every  $\beta>0$, it holds that $\exp[-\beta W^*H_{\mathrm{H}}W]\rhd0\wrt\p_{\rm e}$, where  $W$ is the hole-particle transformation satisfying \eqref{WProp}. Furthemore,  the ground state of  $H_{\rm H}$ is unique and has total spin $S=0$.
\end{Lemm}

\begin{proof}
We outline the proof.   $H_{\rm H}$ can be regarded as the Hubbard Hamiltonian  on the enlarged lattice 
$\varXi=\vLa\sqcup \vLa$\footnote{A similar idea is also used in \cite{PhysRevB.50.6246}.}, 
where the bipartite structure of $\varXi$ is given by
$\varXi=\varXi_1 \cup \varXi_2\ (\varXi_1=\vLa_1\sqcup \vLa_2, \ \varXi_2=\vLa_2\sqcup \vLa_1)$. 
Therefore, the generalized Lieb's theorem \cite{Miyao2012, Miyao2019} can be applied, and $\exp[-\beta W^*H_{\mathrm{H}}W]\rhd0\wrt\p_{\rm e}\ (\beta >0)$  follows.
 From Theorem \ref{pff}, the ground state is unique, and from Lieb's theorem, the total spin of the ground state is given by 
\be
S=\frac{1}{2} \big||\varXi_1|-|\varXi_2|\big|=\frac{1}{2}\big||\vLa_1|+|\vLa_2|-|\vLa_2|-|\vLa_1|\big|=0.
\ee
\end{proof}

Define the reference Hamiltonian as:
\begin{align}
H_{\mathrm{H}}'
=H_{\mathrm{H}} +\omega_0\Np.
\end{align}
$H_{\mathrm{H}}'$ acts in $\h$.

\begin{Lemm}\label{RefHamiP}
$\exp[-\beta W^*H_{\mathrm{H}}'W]\rhd0\wrt\p$ holds  for all $\beta>0$, and thus the ground state of $H_{\mathrm{H}}'$ is unique. 
Furthermore, the ground state has total spin $S=0$.
\end{Lemm}

\begin{proof}
Note that 
$e^{-\beta W^* H_{\mathrm{H}}' W}=e^{-\beta W H_{\rm H} W} \otimes e^{-\beta \omega_0 \Np}$ holds. 
Considering Lemma \ref{HubbardPI} and the fact that $e^{-\beta \omega_0 \Np} \rhd 0$ w.r.t. $L^2_+(\BbbR^{|\vLa|})\ (\beta>0)$, we can apply \cite[Corollary I.8]{Miyao2019} and conclude that $\exp[-\beta W^*H_{\mathrm{H}}'W]\rhd0\wrt\p$ for all $\beta>0$. Thus, the uniqueness of the ground state follows from Theorem \ref{pff}.
The claim for total spin follows immediately from Lemma \ref{HubbardPI}.
\end{proof}

\begin{Lemm}[Positive overlap principle]\label{Overlap2}
Let 
$A $ and $B$ be  positive  self-adjoint operators on 
$
\h
$.
 Let $V_1$ and $V_2$ be unitary operators on $\h$.
 We assume the following: 
\begin{itemize}
\item[{\rm (i)}] $A$ and $B$ commute with the total spin operators $S_{\rm tot}^{(1)}, S_{\rm tot}^{(2)}$ and $S_{\rm tot}^{(3)}$.
\item[{\rm (ii)}] Let 
$V=V_1V_2$.
$\{ e^{-\beta V^*AV} \}_{\beta \ge 0}$ and $\{e^{-\beta V_2^*BV_2}\}_{\beta \ge 0}$ are ergodic w.r.t. $\p$. Hence, the ground state of each of $V^*AV$ and $V_2^*BV_2$ is unique and strictly positive w.r.t. $\p$ due to Theorem \ref{pff}. 
\item[{\rm (iii)}] $V_1$ commutes with $\boldsymbol{S}_{\rm{}tot}^2$.
\item[\rm (iv)] $\inf \mathrm{spec}(A)$ and $\inf \mathrm{spec}(B)$ are eigenvalues of $A$ and $B$, respectively.
\end{itemize}
We denote by $S_A$ (resp. $S_B$) the total spin of the ground state of $A$ (resp. $B$).
Then we have $S_A=S_B$.
\end{Lemm}

\begin{proof}
See \cite[Lemma 4.9]{MIYAO2021168467}.
\end{proof}

\subsubsection*{Proof of Proposition \ref{GTotalSP}}
Taking Theorem \ref{MainPI} and Lemma \ref{RefHamiP} into consideration, we can apply 
 Lemma \ref{Overlap2} with 
$V_1=e^{-L_d}e^{-\im\frac{\pi}{2} N_{\mathrm{p}}}$,  $V_2=W$,  $V=V_1 V_2=\mathcal{U}$,
$A=H$ and $ B= H_{\mathrm{H}}'$.   \qed

\subsection{Proof of Theorems \ref{MainTh1} and \ref{MainTh2}}
So far, we have studied the operator $\bs H\restriction \h$ which restricts $\bs H$ to the $M=0$ subspace $\h$ of $\F_{{\rm e}, 2|\vLa|}\otimes \h_{\rm ph}$.
From the general theory of spin angular momentum, it is shown that if the ground state of $\bs H$ in the $M=0$ subspace is unique and has total spin $S$, then the ground state of $\bs H$ is $2S+1$-fold degenerate by the conservation law of total spin.
In the present case, since $S=0$ follows from Proposition \ref{GTotalSP}, we know that the ground state of $\bs H$ is unique. Consequently, this proves Theorem \ref{MainTh1}. Theorem \ref{MainTh2} follows immediately from Theorem \ref{GSPart}.

The above proof covers $\bs H=\bs H_d$. As discussed in Remark \ref{Reason}, the arguments above are also valid for $\bs H_f$ with appropriate modifications.
\qed

\appendix
\section{Proof of Proposition \ref{Projection3}}\label{PfProjection3}

\subsection{Preliminaries}
In this appendix, we will prove Proposition \ref{Projection3}.
For this purpose, some preparations are necessary.

For each $x\in \vLa$ and $\sigma=\up, \down$, we set
\be
v_{x, \sigma}^-=f_{x, \sigma}^*d_{x, \sigma}, \quad v_{x, \sigma}^+=d_{x, \sigma}^*f_{x, \sigma}.
\ee
Obviously, $B_{x, x;  \sigma}=V v_{x, \sigma}^-$
and additionally it holds that 
\begin{align}
v_{x, \sigma}^- v^+_{x, \sigma}=n_{x, \sigma}^f\overline{n}_{x, \sigma}^d, \ \ v_{x, \sigma}^+v_{x, \sigma}^-=\overline{n}_{x, \sigma}^f n_{x, \sigma}^d,\ (v_{x, \sigma}^-)^2=0, \label{vBase1}\\
[v_{x, \sigma}^-, v^-_{y, \tau}]=0=[v^-_{x, \sigma}, v_{y, \tau}^+], \ \ \ \mbox{$x\neq y$ or $\sigma\neq \tau$}. 
\end{align}

\begin{Lemm}\label{Projection1}
Given $x_1, x_2, y_1, y_2\in \vLa$, set $\delta_{(x_1, x_2), (y_1, y_2)}=\delta_{x_1, y_1} \delta_{x_2, y_2}$.  Then the following hold:
\begin{align}
n_{x,  \sigma}^f v^-_{y, \tau}\overline{n}_{x,  \sigma}^f
&=  \delta_{\sigma, \tau} \delta_{x, y}v_{x, \sigma}^-,  \label{f1}\\
n_{x,  \sigma}^f v_{y, \tau}^+\overline{n}_{x,  \sigma}^f
&=  0,  \label{f3}\\
\overline{n}_{x,  \sigma}^f v_{y, \tau}^+ n_{x,  \sigma}^f
&=\delta_{\sigma, \tau} \delta_{x, y}v_{x, \sigma}^+, \label{f3} \\
\overline{n}_{x,  \sigma}^f v_{y, \tau}^- n_{x,  \sigma}^f
&=0,  \label{f4}\\
n_{x,  \sigma}^f v_{y_1, \tau_1}^- \overline{n}_{x,  \sigma}^f v_{y_1, \tau_2}^+  n_{x,  \sigma}^f
&= \delta_{(y_1, y_2), (x, x)}  \delta_{(\tau_1, \tau_2), (\sigma, \sigma)} n_{x,  \sigma}^f\overline{n}_{x,  \sigma}^d , \label{f5}\\
n_{x,  \sigma}^f v_{y_1, \tau_1}^- \overline{n}_{x,  \sigma}^f v_{y_1, \tau_2}^- n_{x,  \sigma}^f
&= n_{x,  \sigma}^f v_{y_1, \tau_1}^+  \overline{n}_{x,  \sigma}^f v_{y_1, \tau_2}^+ n_{x,  \sigma}^f
= n_{x,  \sigma}^f v_{y_1, \tau_1}^+  \overline{n}_{x,  \sigma}^f v_{y_1, \tau_2}^- n_{x,  \sigma}^f=0, 
\label{f6} \\
\overline{n}_{x,  \sigma}^f v_{y_1, \tau_1}^+  n_{x,  \sigma}^f v_{y_1, \tau_2}^-\overline{n}_{x,  \sigma}^f
&=\delta_{(y_1, y_2), (x, x)}  \delta_{(\tau_1, \tau_2), (\sigma, \sigma)} \overline{n}_{x, \sigma}^f n_{x, \sigma}^d,\label{f7}\\
\overline{n}_{x,  \sigma}^f v_{y_1, \tau_1}^-  n_{x,  \sigma}^f v_{y_1, \tau_2}^- \overline{n}_{x,  \sigma}^f
&=\overline{n}_{x,  \sigma}^f v_{y_1, \tau_1}^+ n_{x,  \sigma}^f v_{y_1, \tau_2}^+\overline{n}_{x,  \sigma}^f
=\overline{n}_{x,  \sigma}^f v_{y_1, \tau_1} ^+n_{x,  \sigma}^f v_{y_1, \tau_2}^-\overline{n}_{x,  \sigma}^f=0.\label{f8}
\end{align}

\end{Lemm}
\begin{proof}
From the anticommutation relations \eqref{CAR1} and \eqref{CAR2}, we know that
\begin{align}
n_{x, \sigma}^f f_{x, \sigma}^*&=f_{x, \sigma}^*, \ \ f_{x, \sigma} n_{x, \sigma}^f=f_{x, \sigma},\ \
f_{x, \sigma}^* \overline{n}_{x, \sigma}^f=f_{x, \sigma}^*,\ \ \overline{n}_{x, \sigma}^f f_{x, \sigma}=f_{x, \sigma},\\
f_{x, \sigma}^* n_{x, \sigma}^f&=0,\ \ n_{x, \sigma}^f f_{x, \sigma}=0,\ \ \overline{n}_{x, \sigma}^f f_{x, \sigma}^*=0,\ \ f_{x, \sigma} n_{x, \sigma}^f=0.
\end{align}
With these, we find that \eqref{f1}-\eqref{f4} are valid. Furthermore, \eqref{f5}-\eqref{f8} can be immediately obtained from \eqref{vBase1}-\eqref{f4}.
\end{proof}

For each 
$X\subset \vLa$ and    ${\bs \vepsilon}=(\vepsilon_x)_{x\in X}\in \{-, +\}^{|X|}$, define
\be
{\bs v}_{X, \sigma}^{\bs \vepsilon}=\prod_{x\in X}v_{x, \sigma}^{\vepsilon_{x}}.
\ee
For $\bs \vepsilon=(\vepsilon_x)_{x\in X}$ defined  by $\vepsilon_x=+$ for all $x\in \vLa$, we shall write  ${\bs v}_{X, \sigma}^{\bs \vepsilon}={\bs v}_{X, \sigma}^{ +}$.
Similarly, for $\bs \vepsilon=(\vepsilon_x)_{x\in X}$ defined  by $\vepsilon_x=-$ for all $x\in \vLa$, let us  denote ${\bs v}_{X, \sigma}^{\bs \vepsilon}={\bs v}_{X, \sigma}^{ -}$.

\begin{Lemm}\label{KeyTech}
Suppose that 
 $Y, Y\rq{}\subset \vLa$  satisfy $|Y| \le |X_d\triangle X_f|$ and  $|Y\rq{}| \le |X_d\triangle X_f|$, respectively.
One obtains 
\begin{align}
&Q_{\bs X} {\bs v}_{Y, \up}^{\bs \vepsilon}{\bs v}_{Y\rq{}, \down}^{\bs \vepsilon\rq{}} P_{\bs X}\no
=&
\begin{cases}
\Big[{\bs v}^{-}_{X_d\setminus X_f, \uparrow}{\bs v}^{-}_{X_d\setminus X_f, \down}\Big]\Big[
{\bs v}^+_{X_f\setminus X_d, \uparrow}{\bs v}^+_{X_f\setminus X_d, \down}
\Big] & \mbox{if $Y=Y\rq{}=X_d\triangle X_f$ and ${\bs \vepsilon}={\bs \vepsilon\rq{}}={\bs \vepsilon_X}$}\\
0 & \mbox{otherwise, }
\end{cases} \label{QVP}
\end{align}
where ${\bs \vepsilon}_{\bs X}=(\vepsilon_x)_{x\in X_d\triangle X_f}$ is given by 
\be
\vepsilon_x=
\begin{cases}
- & \mbox{if $x\in X_d\setminus X_f$}\\
+ & \mbox{if $x\in X_f\setminus X_d$}.
\end{cases}
\ee

\end{Lemm}

\begin{proof}
Use Lemma \ref{Projection1}.
\end{proof}

\subsection{Completion of the proof of Proposition  \ref{Projection3}}

Throughout this proof, set $N=|\bs X|_{\triangle}$.
We can express $P_{\bs X}$ as 
\begin{align}
P_{\boldsymbol{X}}
&=\prod_{\sigma=\up, \down} P^f_{X_f\setminus X_d, \sigma} \overline{P}^f_{X_d\setminus X_f, \sigma} P^f_{X_d\cap X_f, \sigma} \overline{P}^f_{X_d^c\cap X_f^c, \sigma}.
\end{align}
Putting 
\be
P_1=\prod_{\sigma=\up, \down} P^f_{X_d\cap X_f, \sigma} \overline{P}^f_{X_d^c\cap X_f^c, \sigma}, 
\ee
we have $P_{\bs X}=P_{\bs X} P_1$. Similarly, we obtain $ Q_{\bs X}=Q_{\bs X} P_1$.
For each $X\subseteq \vLa$ and $\bs \vepsilon=(\vepsilon_x)_{x\in X}\in \{-, +\}^{|X|}$,  define
\be
\varTheta_{X}^{\bs \vepsilon}=\prod_{x\in X} e^{-\vepsilon_x \im \vPhi_x},\quad \overline{\varTheta}_{X}^{\bs \vepsilon}=\prod_{x\in X} e^{+\vepsilon_x \im \vPhi_x}.
\ee

From  Lemma \ref{KeyTech},  it follows that 
\begin{align}
&(N!)^{-2}Q_{\boldsymbol{X}} V_{\up}^{N}(\VP)V_{\down}^{N}(-\VP) P_{\boldsymbol{X}}\no
&=(N!)^{-2} V^{2N} \sum_{{X\subset \vLa}\atop{|X|=N}} \sum_{
{Y\subset \vLa}\atop{|Y|=N}
}
\sum_{{\bs \vepsilon}, {\bs \vepsilon}\rq{}\in \{-, +\}^{N}} \varTheta_{X}^{\bs \vepsilon}\overline{\varTheta}_{Y}^{\bs \vepsilon'}
Q_{\boldsymbol{X}}
v_{X, \up}^{\bs \vepsilon}v_{Y, \down}^{\bs \vepsilon\rq{}}
P_{\boldsymbol{X}}\no
&=(N!)^{-2} V^{2N} \sum_{X=X_f\triangle X_d} \sum_{Y= X_f\triangle X_d}
\sum_{{\bs \vepsilon}={\bs \vepsilon}_{\bs X}}
\sum_{{\bs \vepsilon}\rq{}={\bs \vepsilon}_{\bs X}
}\varTheta_{X}^{\bs \vepsilon}\overline{\varTheta}_{Y}^{\bs \vepsilon'}
P_1
Q_{\boldsymbol{X}}
v_{X, \up}^{\bs \vepsilon}v_{Y, \down}^{\bs \vepsilon\rq{}}
P_{\boldsymbol{X}}
P_1
\no
&=V^{2N}
\Big[{\bs v}^{-}_{X_d\setminus X_f, \uparrow}{\bs v}^{-}_{X_d\setminus X_f, \down}\Big]\Big[
{\bs v}^+_{X_f\setminus X_d, \uparrow}{\bs v}^+_{X_f\setminus X_d, \down}
\Big] P_1, 
\end{align}
where we have used the fact that if  $X=X_f\triangle X_d,\ Y= X_f\triangle X_d$ and $ {\bs \vepsilon}={\bs \vepsilon}'={\bs \vepsilon}_{\bs X}$, then 
$\varTheta_{X}^{\bs \vepsilon}\overline{\varTheta}_{Y}^{\bs \vepsilon'}=\mathbbm{1}$ holds.

From Duhamel's formula, we get 
\be
e^{-\beta H_0}=\sum_{n=0}^{\infty} G_n(\beta),
\ee
where  $G_0(\beta)=e^{-\beta \omega_0 \Np}$ and 
 \be
 G_n(\beta)=\beta^n \int_{\Delta_n} H_1(s_1)\cdots H_1(s_n) e^{-\beta \omega_0 \Np}d^n {\bs s},\quad H_1(s)=e^{-s \beta \omega_0 \Np}H_1 e^{s\beta \omega_0 \Np}.
 \ee
 Because 
$\lim_{\beta \to +0}(\mathbbm{1} -e^{-s \beta \omega_0 \Np})\vphi=0\ (\vphi\in \h)$ and  $\|\mathbbm{1} -e^{-s \beta \omega_0 \Np}\| \le 2 \ (s\ge 0)$ hold, we find that 
\be
 \lim_{\beta \to +0}H_1(s_1)\cdots H_1(s_n) e^{-\beta \omega_0 \Np} \vphi =H_1^n \vphi \quad (\vphi\in \h), 
\ee
which implies that 
\be
G_n(\beta)=\frac{\beta^n}{n!} H_1^n +o(\beta^n)\quad (\beta\to +0). \label{SimpNinaru}
\ee
Since 
$Q_{\bs X}
H_1(s_1)\cdots H_1(s_n) e^{-\beta \omega_0 \Np} 
P_{\boldsymbol{X}}=0$ holds for $n<2N$, one obtains 
\be
Q_{\bs X} e^{-\beta H_0}P_{\bs X} =G_{2N}(\beta)+O(\beta^{2N+1}). \label{2Nkieru}
\ee
Besides, one finds that 
\begin{align}
Q_{\bs X} H_1^{2N} P_{\bs X}=Q_{\bs X} \mathbbm{V}^{2N} P_{\bs X} =\binom{2N}{N} Q_{\boldsymbol{X}} V_{\up}(\VP)^NV_{\down}(-\VP)^{N} P_{\boldsymbol{X}} \label{UDdake}
\end{align}
holds. Then, by setting 
\begin{align}
K_{\boldsymbol{X}}
=\Big[{\bs v}^{-}_{X_d\setminus X_f, \uparrow}{\bs v}^{-}_{X_d\setminus X_f, \down}\Big]\Big[
{\bs v}^+_{X_f\setminus X_d, \uparrow}{\bs v}^+_{X_f\setminus X_d, \down}
\Big] P_1,
\end{align}
we see from \eqref{QVP},  \eqref{SimpNinaru}, \eqref{2Nkieru} and \eqref{UDdake} that 
\begin{align}
Q_{\boldsymbol{X}} e^{-\beta H_0} P_{\boldsymbol{X}}
\overset{\eqref{SimpNinaru}, \eqref{2Nkieru}}{=}&\frac{\beta ^{2N}}{(2N)!}Q_{\boldsymbol{X}} H_1^{2N} P_{\boldsymbol{X}}
+o(\beta^{2N}) \no
\overset{\eqref{UDdake}}{=}&\frac{\beta^{2N}}{(2N)!} \binom{2N}{N}  Q_{\boldsymbol{X}} V_{\up}(\VP)^{N}V_{\down}(-\VP)^{N} P_{\boldsymbol{X}} 
+o(\beta^{2N}) \no
\overset{\eqref{QVP}}{=}&\frac{\beta ^{2N}}{(2N)!}\binom{2N}{N}   (N!)^2V^{2N}K_{\boldsymbol{X}}
+o(\beta^{2N})
\no
=\ \ &\beta^{2N}V^{2N}K_{\boldsymbol{X}}
+o(\beta^{2N}).
\end{align}
Therefore, we have
\begin{align}
P_{\boldsymbol{X}} e^{-\beta H_0} Q_{\boldsymbol{X}} e^{-\beta H_0} P_{\boldsymbol{X}} 
&=\left\{\beta^{2N}V^{2N}K_{\boldsymbol{X}}^*+o(\beta^{2N})\right\} \left\{\beta^{2N}V^{2N}K_{\boldsymbol{X}}
+o(\beta^{2N})\right\}\no
&=\beta^{4N}V^{4N}K_{\boldsymbol{X}}^*K_{\boldsymbol{X}} +o(\beta^{4N}). \label{PeP}
\end{align}
According to \eqref{vBase1}, we know that $K_{\boldsymbol{X}}^*K_{\boldsymbol{X}}=\Ex_{\boldsymbol{X}}$, so it follows that
\begin{align}
P_{\boldsymbol{X}} e^{-\beta H_0} Q_{\boldsymbol{X}} e^{-\beta H_0} P_{\boldsymbol{X}}
&=\beta^{4N}V^{4N}\Ex_{\boldsymbol{X}}
+o(\beta^{4N}).
\end{align}
\qed

\section{Proof of Proposition \ref{Connect}}\label{PfConnect}

\subsection{Classification of electron configurations}
In this appendix, we prove Proposition \ref{Connect}. This proposition looks simple, but its proof is, unexpectedly, rather complicated.

For the sake of subsequent discussion, we classify the electron configurations as follows.

\begin{Def}\upshape 
\begin{itemize}
\item[(i)] 
An electron configuration ${\bs X}=(X_d,  X_f)\in \mathscr{C}$ is said to be {\it $d$-dominated} if it satisfies $X_d \supseteq X_f$.
 We denote by $\mathscr{C}_d$ the set of all $d$-dominated electron configurations:
$
\mathscr{C}_d= \{{\bs X}=(X_d, X_f)\in \mathscr{C} \, :\,   X_d \supseteq X_f\}
$(See Fig. \ref{Fig1}).
\begin{figure}[t]
\centering
 \includegraphics[scale=0.6]{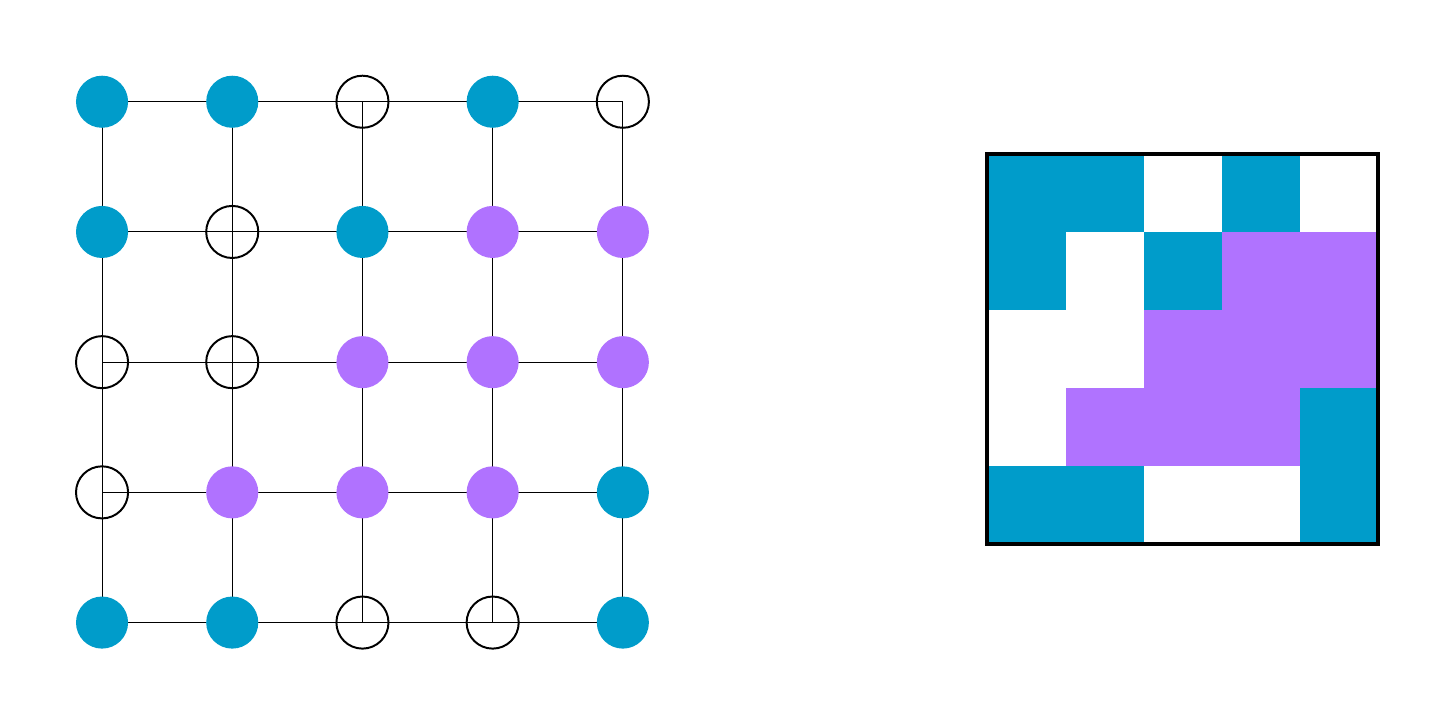}
\caption{The left figure is an example of the $d$-dominated electron configuration. See Fig. \ref{Fig0} for the meaning of the colors of the sites. The right figure is a simplified representation of the situation in the left.}
\label{Fig1}
\end{figure}
\item[(ii)] 
An electron configuration ${\bs X}=(X_d,  X_f)\in \mathscr{C}$ is said to be {\it $f$-dominated} if it satisfies $X_f \supseteq X_d$.
We denote by $\mathscr{C}_f$ the set of all $f$-dominated electron configurations:
$
\mathscr{C}_f= \{{\bs X}=(X_d, X_f)\in \mathscr{C} \, :\,   X_f \supseteq X_d\}
$(See Fig. \ref{Fig2}).

\item[(iii)]
An electron configuration ${\bs X}=(X_d,  X_f)\in \mathscr{C}$ is said to be {\it $(d, f)$-disjoint} if it satisfies $X_d\cap X_f=\varnothing$. We denote by $\mathscr{C}_{d, f}$ the set of all  $(d, f)$-disjoint electron configurations:
$\mathscr{C}_{d, f}= \{{\bs X}=(X_d, X_f) \in \mathscr{C}\, :\,  X_d \cap X_f=\varnothing\}$(See Fig. \ref{Fig2}).

 \begin{figure}[t]
\centering
 \includegraphics[ scale=0.6]{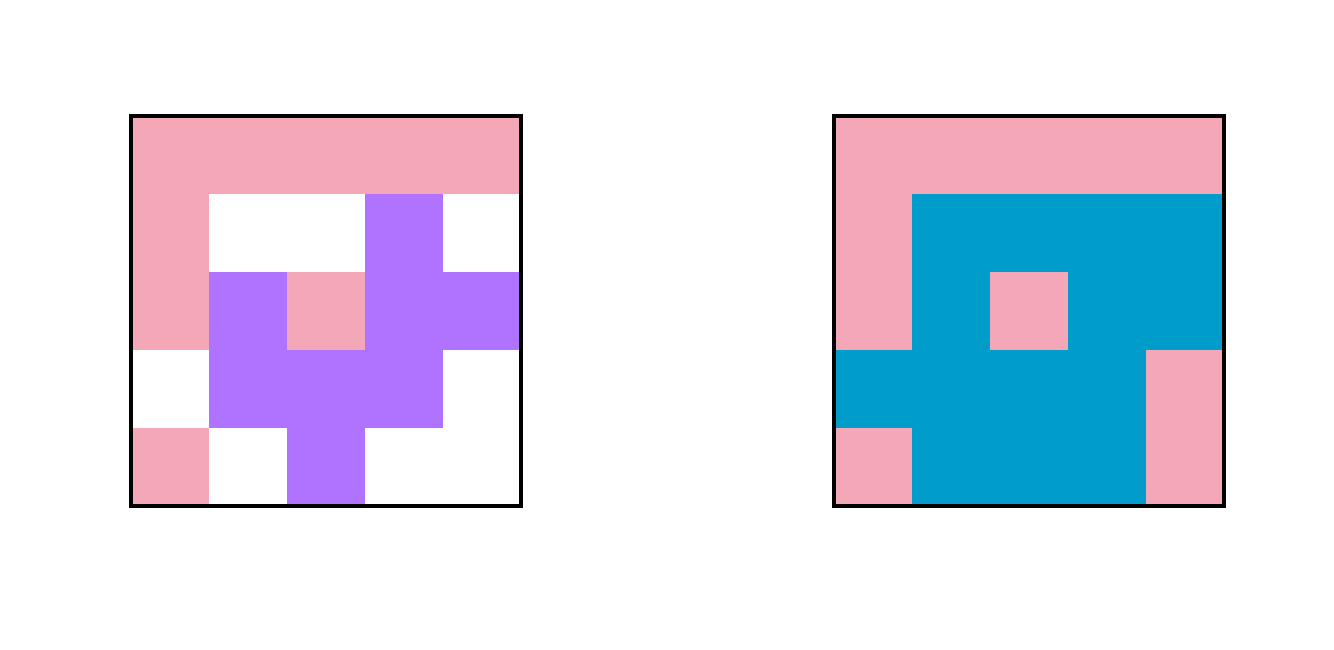}
\caption{The left figure shows an example of the $f$-dominated electron configuration.
The right figure illustrates an example of the $(d, f)$-disjoint electron configuration.}
\label{Fig2}
\end{figure}
\end{itemize}
\end{Def}

\subsection{Simplification of electron configurations not in $\mathscr{C}_f$}
In this subsection, we examine the properties of non $f$-dominated electron configurations.

We first introduce the  term as follows to facilitate the following discussion.
\begin{Def}
\rm For each  $\bs X=(X_d, X_f)\in \mathscr{C} \setminus \mathscr{C}_f$, we define the new electron configuration $S \bs X$ by $S \bs X=(X_d\cap X_f, X_d\cup X_f)$. Apparently $S\bs X\in \mathscr{C}_f$. We call $S\bs X$ a {\it simplification} of $\bs X$. See Fig. \ref{Fig3}
\end{Def}

\begin{figure}[h]
\centering
 \includegraphics[ scale=0.6]{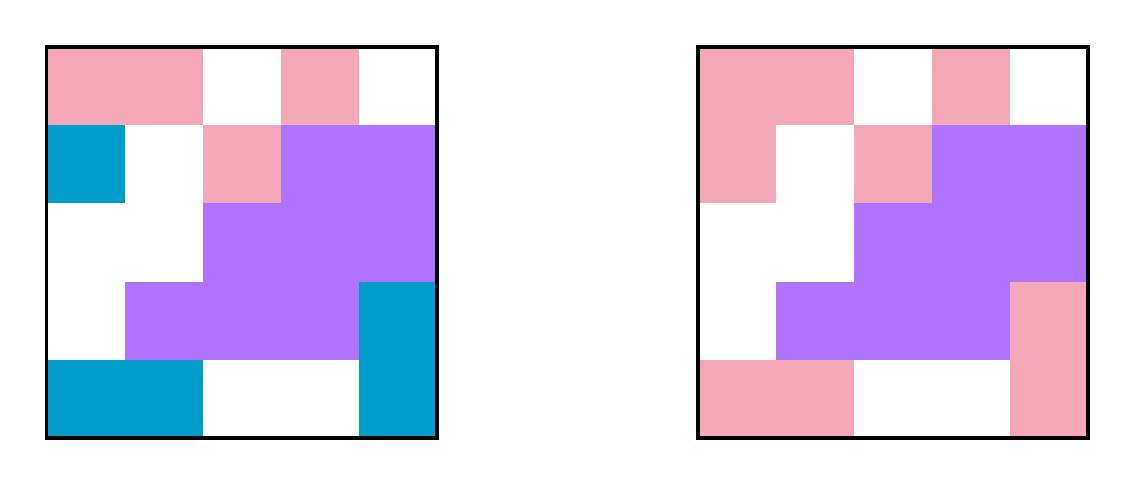}
\caption{The right figure depicts the simplification $S\bs X$ of the electron configuration $\bs X$ represented in the left figure.}
\label{Fig3}
\end{figure}
The purpose of this subsection is to prove the following lemma:

\begin{Lemm}\label{Simplified}
Let ${\bs X}=(X_d, X_f)\in \mathscr{C}\setminus \mathscr{C}_f$. We set 
  $m=| X_d\setminus X_f|$. Then there exists a path  ${\bs p}={\bs X}_1\cdots {\bs X}_{m+1}$ connecting $\bs X$ and its simplification $S {\bs X}$ such that each $\{\bs X_i, \bs X_{i+1}\}$ is a $(d, f)$-edge $(i=1, \dots, m)$.

\end{Lemm}
\begin{proof}
Take $x\in X_d\setminus X_f$, arbitrarily. We define the electron configuration 
${\bs X_m}$  as  ${\bs X}_m=(X_d\setminus \{x\},  X_f\cup \{x\})$.
Then $\{\bs X_m, \bs X_{m+1}\}$ is a $(d, f)$-edge. For each 
$j=2,\ldots, m-1$, define ${\bs X}_{m-j}=(X_{d,  m-j+1}\setminus \{x_{j}\},  X_{f,  m-j+1}\cup\{x_{j}\})\ (x_j\in X_{d,  m-j+1}\setminus X_{f,  m-j+1})$. Then, we see that each  $\{\bs X_{m-j}, \bs X_{m-j+1}\}$ is a $(d, f)$-edge.
Defining $\bs p={\bs X}_1\cdots {\bs X}_{m+1}$ shows the assertion of Lemma \ref{Simplified}.
\end{proof}

\subsection{Properties of electron configurations in $\mathscr{C}_f$}
In this subsection we prove the following lemma for $f$-dominated electron configurations.

\begin{Lemm}\label{XtoF}
For any $\bs X\in \mathscr{C}_f$, there exists a path connecting $\bs X$ and $\bs F$.
\end{Lemm}

\begin{proof}
Let $G_d$ be the graph generated by the hopping matrix introduced in Section \ref{Sec1}.
From the assumption \hyperlink{A2}{\bf (A. 2)}, $G_d$ is a connected graph.

For each $Z\subset \vLa$, define its boundary $\partial Z$ by $\partial Z=\{z\in Z\, :\,  \exists x\in Z^c \ \mbox{s.t. $t_{x,z} \neq 0$}\}$.
Set $Y=\vLa\setminus (X_d\cup X_f)$ 
For each 
$z\in \partial Y$ and $w\in \partial (X_d\cap X_f)$, we set
\be
P_{\bs X}(z, w)=\{p: \mbox{path from $z$ to $w$ satisfying $p\subset X_f\setminus X_d$}\},
\ee
where for a path $p=x_1\cdots x_n$ and $A\subset \vLa$, we denote by $p\subset A$ if $\{x_1, \dots, x_n\}\subset A$ holds.
Furthermore, put
\be
N_{\boldsymbol{X}}(z,w)=\min_{p\in P_{\bs X}(z, w)} |p|.
\ee
Fig. \ref{Fig4} depicts the situation described above.
We choose $z_1\in \partial Y$ and $w_1\in \partial(X_d\cap X_f)$ so that $P_{\bs X}(z_1, w_1) \neq \varnothing$.
Let $n_1=N_{\boldsymbol{X}}(z_1,w_1)$ and denote one of the shortest paths  in $P_{\bs X}(z_1, w_1)$ by $p=u_0u_1\cdots u_{n_1}\ (u_0=w_1, u_{n_1}=z_1)$.  
\begin{figure}[t]
\centering
 \includegraphics[ scale=0.9]{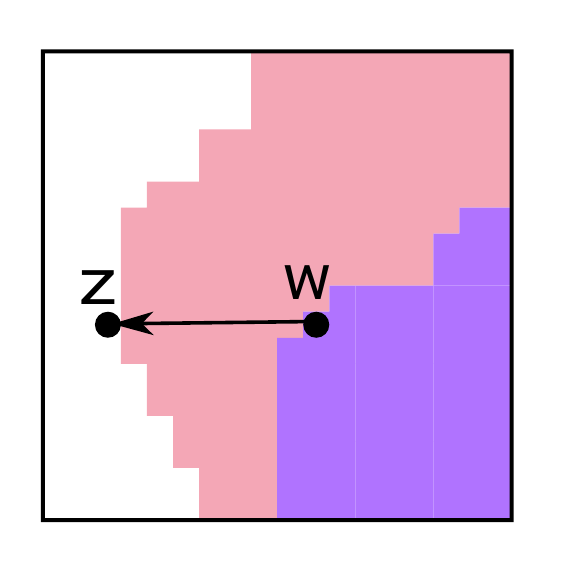}
\caption{$z$ is a site at the boundary of $Y$, and $w$ is a site at the boundary of $X_d\cap X_f$. The arrow represents the shortest path connecting $z$ and $w$.}
\label{Fig4}
\end{figure}
Define the electron configurations $\boldsymbol{Z}_1,\ldots,\boldsymbol{Z}_{n_1}$ by
\begin{align}
\boldsymbol{Z}_0&=\boldsymbol{X}, \\
Z_{j, d}&=(Z_{j-1, d} \setminus \{u_{j-1}\}) \cup \{u_j\},\ 
Z_{j, f}=Z_{j-1, f}
\ (j=1,\ldots,n_1).
\end{align}
We readily confirm that each 
$\{\bs Z_{j-1}, \bs Z_j\}\ (j=1, \dots, n_1)$ is a $d$-edge.
Next, we define the electron configuration $\boldsymbol{Z}_{n_1+1}$ as
\be
Z_{n_1+1, d}=Z_{n_1, d} \setminus \{z_1\},\quad Z_{n_1+1, f}=Z_{n_1, f}\cup\{z_1\}.
\ee
Then $\{\bs Z_{n_1}, \bs Z_{n_1+1}\}$ is a $(d, f)$-edge. 
Additionally, the following hold:
\be
\vLa\setminus (Z_{n_1+1, d} \cup Z_{n_1+1, f})=Y\setminus \{z_1\},\quad Z_{n_1+1, d} \cap Z_{n_1+1, f}=(X_f\cap X_d)\setminus \{w_1\}.
\ee
In this way, we can construct a path $\bs p_1=\bs Z_0\cdots \bs Z_{n_1+1}$ connecting $\bs Z_0=\bs X$ and $\bs Z_{n_1+1}$.

Next, take $z_2\in\partial (\vLa\setminus(Z_{f, n_1+1}\cup Z_{d, n_1+1})) $ and $ w_2\in \partial (Z_{f, n_1+1}\cap Z_{d, n_1+1})$ so that $P_{\bs Z_{n_1+1}}(z_2, w_2) \neq \varnothing$.
Let $n_2=N_{\boldsymbol{Z}_{n_1+1}}(z_2,w_2)$ and let $p_2$ be one of the shortest paths in $P_{\bs Z_{n_1+1}}(z_2, w_2)$.
In a similar way as above, we can construct electron configurations $\boldsymbol{Z}_{n_1+2},\ldots,\boldsymbol{Z}_{n_1+n_2+2}$
such that each 
$\{\bs Z_{j-1}, \bs Z_j\}\ (j=1, \dots, n_1)$ is an edge and the following hold:
\begin{align}
\vLa\setminus (Z_{n_1+n_2+1, d} \cup Z_{n_1+n_2+2, f})&=Y\setminus \{z_1, z_2\},\\
 Z_{n_1+n_2+1, d} \cap Z_{n_1+n_2+1, f}&=(X_f\cap X_d)\setminus \{w_1, w_2\}.
\end{align}
From the electron configurations obtained in this way,
we can  construct a path $\bs p_2$ connecting $\bs Z_{n_1+1}$ and $\bs Z_{n_1+n_2+2}$ by putting $\bs p_2={\bs Z}_{n_1+1}\cdots \bs Z_{n_1+n_2+2}$.

 Repeating this procedure until there are no more elements in $Y$ ($k=|Y|$ times), we obtain a sequence of paths $\bs p_1, \dots, \bs p_k$.  It can be seen that their composition $\bs p=\bs p_k\circ \cdots \circ \bs p_1$ is a path connecting $\bs X$ and $\bs F$.
\end{proof}

\subsection{Completion of the proof of Proposition \ref{Connect}}
We split the proof into two parts.

{\bf Step 1.}
Let  ${\bs X} \in \mathscr{C}$.
Consider the following two cases: (i) $\bs X\in \mathscr{C}_f$; (ii)  $\bs X\in \mathscr{C}\setminus \mathscr{C}_f$.
In case (i),  from Lemma \ref{XtoF}, there exsits a path connecting $\bs X$ and $\bs F$.
On the other hand, in case (ii), from Lemma  \ref{Simplified}, there exists a path $\bs p_1$ connecting $\bs X$ and its simplification $S\bs X$; furthermore, by Lemma \ref{XtoF}, there exists a path $\bs p_2$ connecting $S\bs X$ and $\bs F$.
Then, by combining these two paths, we obtain a path connecting $\bs X$ and $\bs F$.
Summarizing the above, we see that in both cases (i) and (ii), there exists a path $\bs p_{\bs X\to \bs F}$ connecting $\bs X$ and $\bs F$.

{\bf Step 2.}
Take $\bs X, \bs Y\in \mathscr{C}$, arbitrarily. From {\bf Step 1}, two paths exist $\bs p_{\bs X\to \bs F}$ connecting $\bs X$ and $\bs F$ and $\bs p_{\bs Y\to \bs F}$ connecting $\bs Y$ and $\bs F$. Then, by setting $\bs p=\bs p_{\bs X\to \bs F} \circ \bs p_{\bs F\to \bs Y}$, we can construct a path   $\bs p$  connecting $\bs X$ and $\bs Y$,  where $\bs p_{\bs F\to \bs Y}$ is the path obtained by reversing the order of $\bs p_{\bs Y\to \bs F}$, connecting $\bs F$ and $\bs Y$: $\bs p_{\bs F\to \bs Y}=\bs p_{\bs Y\to \bs F}^{-1}$.
\qed

\end{document}